%% file: main.tex
\documentclass[10pt,journal,letterpaper]{IEEEtran}

\usepackage{xspace,xcolor,colortbl} 
\usepackage{amsmath,amsthm,amsfonts}
\usepackage{hyperref} %Might need to remove for final submission
\usepackage[nameinlink,capitalise]{cleveref}
\usepackage{graphicx}
\usepackage[inline]{enumitem}
\usepackage{float}
\usepackage{cite}
\usepackage[noend]{algpseudocode}

\newtheorem{theorem}{Theorem}
\newtheorem{lemma}{Lemma}
\newtheorem{corollary}{Corollary}

\theoremstyle{definition}
\newtheorem{definition}{Definition}

\DeclareMathOperator*{\argmax}{arg\,max}

\DeclareMathOperator*{\polylog}{polylog}

\algnewcommand{\IfThenElse}[3]{% \IfThenElse{<if>}{<then>}{<else>}
  \State \algorithmicif\ #1\ \algorithmicthen\ #2\ \algorithmicelse\ #3}

\newcommand{\Fin}{\ensuremath{\mathcal{F}^{(in)}}}
\newcommand{\Fout}{\ensuremath{\mathcal{F}^{(d)}}}

\renewenvironment{proof}{\begin{IEEEproof}}{%
    \end{IEEEproof}
}

\hypersetup{colorlinks={true},linkcolor={black},citecolor=black}
\hypersetup{colorlinks={false}}

\usepackage{multirow}
\usepackage{colortbl}
\usepackage{hhline}

\begin{document}

\title{Local Fast Rerouting with Low Congestion:\\A Randomized Approach}

\author{Gregor~Bankhamer,~Robert~Els\"asser and Stefan~Schmid\thanks{
Gregor Bankhamer and
Robert Els\"asser are affiliated with the Department of Computer Sciences of the University of Salzburg, Austria.
Stefan Schmid is affiliated with the Faculty of Computer Science of the University of Vienna, Austria and with TU Berlin, Germany.
We would like to thank our anonymous reviewers for their helpful remarks that helped us to improve our paper.
This research was supported in part by the 
Vienna Science and Technology Fund (WWTF)
under grant number ICT19-045, 2020-2024 (project WHATIF) and by the European Union’s Horizon 2020 research and innovation programme under Grant Agreement no. 824115 (HiDALGO).}}

\maketitle

\sloppy

\begin{abstract}
	\input{abstract}
\end{abstract}

\input{introduction.tex}
\input{model.tex}
\input{destination.tex}
\input{interval.tex}
\input{hop.tex}

\input{simulations.tex}

\input{conclusion.tex}

\input{appendix.tex}

\bibliographystyle{IEEEtran}
\bibliography{IEEEabrv,references}

\end{document}

%% file: abstract.tex
Most modern communication networks 
include fast rerouting mechanisms, implemented entirely in the data plane,
to quickly recover connectivity after link failures.
By relying on \emph{local} failure information only, 
these data plane mechanisms provide very fast 
reaction times, but at the same time introduce
an algorithmic challenge in case of multiple link failures: 
failover routes need to be
robust to additional but locally unknown failures
downstream. 

This paper presents local fast rerouting
algorithms which not only provide a high degree 
of resilience against multiple link failures, 
but also ensure a low congestion on the resulting failover
paths. We consider a randomized approach and
focus on networks which are highly connected before the failures occur. 
Our main contributions are three simple algorithms which come
with provable guarantees and provide interesting resilience-load tradeoffs,
significantly
outperforming \emph{any} deterministic fast rerouting algorithm with high probability.

%% file: introduction.tex
\section{Introduction}

Emerging applications, e.g., in the context of industrial, tactile or 5G networks,
come with stringent latency and 
dependability requirements. To meet such requirements,
Fast Re-Route (FRR) mechanisms have been specified for many 
networks~\cite{ipfrr,mplsfrr,offrr,srfrr}:
\emph{local} failover mechanisms in the data plane
which avoid the time-consuming advertisement and collection of 
failure information and re-computation of routes
in the control plane~\cite{isis,rfc2328}. 
Rather, these mechanisms rely on a \emph{pre-defined} logic, 
 often implemented in terms of conditional
failover rules~\cite{offrr}. 
For example, wide-area networks often use IP Fast Reroute~\cite{ipfrr} or MPLS~\cite{mplsfrr}
Fast Reroute to deal with failures on the data plane, 
the Border Gateway Protocol (BGP) uses on BGP-PIC~\cite{filsfils2011bgp} for quickly
rerouting flows, many data centers use Equal Cost MultiPath (ECMP)~\cite{kabbani2014flowbender} which provides automatic 
failover to another shortest path, and 
Software-Defined Networks (SDNs)
provide FRR functionality in terms of OpenFlow fast-failover
groups~\cite{offrr}, among many others~\cite{franccois2014topology}.

However, while FRR mechanisms are attractive 
and widely used to deal with single failures, they
introduce an \emph{algorithmic challenge} in the presence
of \emph{multiple} link failures, as they are 
common in large networks such as datacenter
and Internet networks~\cite{gill2011understanding,concur2,elhourani2014ip}:
rerouting decisions need to be made based on \emph{incomplete} information
about the failure scenario, and in particular, 
about failures \emph{downstream}.
The problem becomes particularly challenging if 
the rerouted flows should not only preserve connectivity
under failures but also a low \emph{load}, an important
criteria in practice: congested routes threaten
dependability and indeed, congestion is a main concern
of any traffic engineering algorithm.

Recently, a series of negative results have been obtained
on what can be achieved using \emph{deterministic} fast rerouting
algorithms (e.g.,~\cite{podc-ba,foerster2020feasibility}). In particular,
it has been shown that even on networks which are still 
highly-connected after failures, 
the congestion resulting from \emph{any} deterministic local fast failover
algorithm is bound to be high in the worst case, i.e., 
\emph{polynomial}
in the number of link failures~\cite{opodis13shoot,borokhovich2018load}. 

This paper initiates the study of \emph{randomized} algorithms to 
provide high resiliency and low congestion at the same time.
In particular, we  show that using a randomized approach,
the congestion can be reduced from polynomial to polylogarithmic, with high probability, hence breaking deterministic congestion lower bounds.

\subsection{Model in a Nutshell}

In a nutshell, 
we consider the fundamental problem of congestion-minimal fast rerouting
on a complete undirected network $G=(V,E)$, where each pair of the $n$ nodes (e.g., switches, 
routers, or hosts) is directly connected (i.e., the network forms a clique). 
Such complete networks are typically studied in the related work 
and can be seen as an approximation of highly-connected networks as they arise, e.g., in the context of datacenters.

The network links (henceforth called \emph{edges}) of $G$ are subject
to multiple concurrent failures, determined by an  adversary
and the goal is to pre-define local failover rules
for the different nodes $V$ such that traffic is rerouted to the destination
while balancing the network \emph{load}.
A failover rule is essentially a \emph{match-action} 
forwarding rule which not only \emph{matches} certain header fields
of the arriving packet (e.g., the IP destination address),
but which can also be conditioned on the link failures incident to a given node  $v \in V$,
thereby specifying for which packets the rule is triggered;
the \emph{action} part then defines to which link the packet needs to be forwarded
accordingly. 
These rules are static, i.e., the routing table is not allowed to be updated 
during the whole routing procedure.
To asses the performance of our protocols, we revisit the challenging (and practically relevant~\cite{wu2012ictcp,handley2017re}) 
in-cast scenario where
$n-1$ sources inject one indefinite flow each to a single destination $d$ which is known to the adversary~\cite{infocom19casa,opodis13shoot,borokhovich2018load,icnp19}. In our empirical analyses we additionally consider the so called gravity model \cite{gravitymodel}, and evaluate adaptions of our algorithms in the Clos fat-tree topology against state-of-the-art approaches.

\subsection{The Deterministic Case Lower Bound}

The authors of \cite{opodis13shoot} showed that deterministic failover algorithms
are bound to result in a high load even in case of an initially completely connected network
which is still highly connected after the failures~\cite{opodis13shoot}.
The proof has been generalized further by 
Pignolet et al.~in~\cite{borokhovich2018load}.
More specifically, 
the authors showed that: (1) 
when only relying on destination-based failover rules (i.e., rules
which can only match the IP destination of a packet), 
an adversary can always induce a load of $\Omega(\varphi)$ at some edge 
by cleverly failing $\varphi$ edges;
(2) when failover rules can also depend on the 
source address, an edge load of $\Omega(\sqrt{\varphi})$ 
can be achieved, when failing $\varphi$ many edges. 

When considering the node load only, this bound can be extended and accounts for further information that may be used by the routing rules. Particularly, if we require that some packet starting from node $v$ takes the same path under the same set of underlying edge failures (i.e., the packets' paths are oblivious and may not change depending on the other traffic moving around the network), a node load of $\Omega(\sqrt{\varphi})$ can be generated by the adversary.  Note that this extension allows for including the hop counter inside the routing rule without weakening the result of the lower bound.

\subsection{Our Results}

The main contribution of this paper are three randomized
fast rerouting algorithms which not only provide a high
resilience to multiple link failures but also
an exponentially lower load than any possible deterministic algorithm.

We present three failover strategies. 
Assuming up to $\varphi= O(n)$ edge failures, 
the first algorithm ensures that a load of $O(\log n \log \log n)$ is not exceeded at most nodes, while the remaining $O(\polylog n)$\footnote{By $\polylog n$, we denote the family of functions in $n$ which lie in $O(\log^p n)$ for any constant $p> 0$.} nodes reach a load of
at most $O(\polylog n)$. As we consider randomized approaches, we require the above statement to hold \emph{with high probability} \footnote{We use the well established notion of \emph{with high probability}, or w.h.p.~, to denote probability of at least $1-n^{-\Omega(1)}$.}.
The second approach we present reduces the edge failure resilience to $O(n/\log n)$, however it is purely \emph{destination-based} and achieves a congestion of only $O(\log n \log \log n)$ at \emph{any} node w.h.p.
Finally, by assuming that the nodes do have access to $\polylog n$ bits of shared information, which are not known to the adversary, the node load can be reduced even further. That is, a maximum load of only $O(\sqrt{\log n})$ occurs at \emph{any} node w.h.p. 
All three strategies ensure loop-freedom w.h.p. \cite{clad2014disruption} and avoid packet reorderings (i.e., all packets of the same flow are forwarded along the same path).

While our focus lies on complete networks, which constitute a major open problem in the literature today, we show how our first two protocols may be adapted to the widely used Clos datacenter topology  \cite{clos,singh2015jupiter}. More precisely, we consider the Clos topology with 3 layers sometimes also simply referred to as fat-tree topology. Besides datacenters, such fat-tree topologies are also employed in some HPC systems, e.g. in the tier-0 supercomputer SuperMUC-NG\footnote{\url{https://doku.lrz.de/display/PUBLIC/SuperMUC-NG}}. We then report on empirical insights obtained through simulations and compare our approaches to other state-of-the-art failover protocols \cite{infocom19casa,DBLP:journals/ton/ChiesaNMGMSS17}. The extension of our protocols to general topologies remains an open problem. However, this may be possible with the help of network decompositions based on
spanning arborescences: 
it is known that any $k$-connected graph can be spanned by $k$ arc-disjoint arborescences~\cite{edmonds1973edge} (such a decomposition can be computed efficiently~\cite{bhalgat2008fast}),
which enables a loop-free resilient routing~\cite{Chiesa2014}.  
The idea is then to use our approach to balance flows across arborescences. Finally, note that our third protocol is mostly of theoretical interest: extending it to general topologies would require to compute a set of Hamilton cycles for each such topology, which is  NP-complete \cite{GareyJ79}. 

\subsection{Further Related Work}
\label{sec:related}

Link failures are the most common failures
in communication networks~\cite{frr-survey,link-failures-ip-backbone,failures-uninett}
and
it is well-known that 
ensuring connectivity via the control plane
can be slow~\cite{ensure-dconn-nsdi13,podc-ba},
even if it is centralized~\cite{Yang14}
or based on link reversal~\cite{gafni-lr,welch2011link}; 
it may also introduce undesirable transient behavior, such as a high loop ratio \cite{clad2014disruption}. 
Data plane based failover mechanisms which do not require
table reconfigurations can be orders of magnitudes faster~\cite{podc-ba}
but are algorithmically challenging as routing tables need to be precomputed
without knowledge of failures.
For a general overview on failover mechanisms in the data plane, we refer 
to the recent survey by Chiesa et al.~\cite{chiesa2020fast}, referencing over
two hundred papers on the topic, several of them published at IEEE/ACM Transactions
on Networking~\cite{menth2009resilience,ref28,cohen2009maximizing,qiu2010local,r38,ref11,DBLP:journals/ton/ChiesaNMGMSS17,clad2013graceful,kvalbein2008multiple,cho2011independent,ref27,elhourani2016ip,borokhovich2018load}.

Except for one model, we are interested in fast failover mechanisms in the data plane
which do not require the modification of packet headers: while the modification of packet headers
can simplify ensuring connectivity (e.g., by carrying failure information in the header)~\cite{fcp,elhourani2014ip},
packet header rewriting typically comes with overheads and may even be infeasible~\cite{podc-ba,plinko-full,robroute16infocom}. Furthermore, while several interesting heuristics have been proposed 
in the literature~\cite{Yang14}, 
we are concerned with mechanisms which come with formal 
(probabilistic) performance guarantees.   
In particular, we are interested in scenarios in which \emph{multiple} links can fail simultaneously; that is, in addition to ensuring traditional properties such as a perfect protection ratio~\cite{francois2005evaluation} (i.e., ensuring resilience against any single failure), we aim to preserve connectivity and low load even under a large number of failures. 

Prior work already derived several fundamental results on the feasibility of
preserving connectivity using deterministic local fast failover mechanisms, in the presence of multiple failures and on the routing level, in different settings. 
In particular, Feigenbaum et al.~\cite{podc-ba} 
proved that it is not possible to achieve a  \emph{perfect (static) resilience}
in arbitrary networks using deterministic local fast rerouting and without header 
rewriting: 
it is impossible to define failover rules such that connectivity
is preserved on the routing level as long as the network is physically
connected. These results were recently extended 
by Foerster et al., who derive more general negative and positive results, 
also considering planar graphs in more details~\cite{foerster2020feasibility}. 
In~\cite{DBLP:journals/ton/ChiesaNMGMSS17}, 
Chiesa et al.~conjecture~\cite{icalp16,DBLP:journals/ton/ChiesaNMGMSS17} that is at least always 
possible to deterministically achieve what they call \emph{ideal resilience}:
unlike perfect resilience which requires connectivity on the routing level as long
as the underlying  \emph{arbitrary} network is connected, 
ideal resilience focuses on $k$-(edge-)connected networks and requires connectivity
on the routing level as long as there are at most $k-1$ 
link failures. Today, it is still unknown whether this conjecture holds in general, however,
at least it has been proved true for several special graph classes as well as for
scenarios with at most $k/2$ failures~\cite{Chiesa2014}.
Furthermore, Chiesa et al.~\cite{icalp16,DBLP:journals/ton/ChiesaNMGMSS17} showed that
ideal resilience can be achieved using randomized algorithms, 
by routing along precomputed spanning arborescences and by switching to a random
alternative arborescence when encountering a failure \cite{bhalgat2008fast}.

The papers discussed above primarily focus on preserving connectivity, and much less is known
about the design of failover algorithms which also account for load. More specifically, while there exist several results on deterministic algorithms for the incast scenario considered in this paper~\cite{dsn19,infocom19casa,foerster2020feasibility,podc-ba}, we are the first to study randomized algorithms and we show that the resulting load can be significantly better than the deterministic lower bound. Motivated by our empirical results on the Clos fat-tree topology (cf.~Section \ref{sec:simul}) described in this paper, in a follow-up conference submission we theoretically analyzed an adapted version of the \emph{Intervals} protocol from Section \ref{sec:interval}, see \cite{BES21}.

\noindent \emph{Bibliographic note.} A preliminary version (without 
all technical details) was presented at IEEE ICNP 2019~\cite{icnp19}. 

%% file: model.tex
\section{Model}
\label{sec:model}
We model the communication network as a complete undirected graph $G=(V,E)$ 
where the nodes $V$ represent the switches or routers which need to be configured 
with static forwarding rules
(in this paper sometimes also simply called routing rules)
and where the links $E$ can fail. 

In particular, we consider three different models, i.e., types of (match-action) \emph{rulesets}, 
which are of increasing power depending on the information
on which the rules can depend (the match part) and the information which they can
change in the packet header (the action part):
\begin{itemize}
	\item \textbf{Destination address:} Rules can only match the destination address (e.g.,
	the IP destination) of the packets. Packets cannot be modified.
	
	\item \textbf{Hop count:} In addition to the destination address, rules can 
	match the hop count: how far the packets have travelled so far.
	
	\item \textbf{Hop count modification:} One of our algorithms (\emph{Shared-Permutations} -- see \cref{sec:hop}) can additionally modify the hop count to arbitrary $O(\log n)$ bit values. Note that this rule is not needed for the first two algorithms we develop.
\end{itemize}
In addition to header information, the matching part of each above rule may also depend on the 
\emph{local link failures}, the link failures incident to the given node $v$,
but not on other remote failures which are not known at the configuration time.
While the first model is the standard model used by routers and supported generally,
the latter two models require software-defined switches or routers, e.g., based
on OpenFlow \cite{offrr}. 

We consider an adversarial model and assume that the link failures are chosen 
by an  adversary. 
More specifically, we assume that the failover rules are generated
by a randomized algorithm (or rely again on a hash function which matches the hop count),
and that the adversary is \emph{oblivious}: it knows the failover protocol including the 
used probability distributions,  but not the generated random values nor the resulting loads.
(Later in this paper we briefly discuss even stronger adversaries.)

In order to assess the performance of our protocols, we consider the \emph{all-to-one} traffic pattern,
in which each node $v \in V \setminus \{d\}$ sends a single \emph{flow} towards some common destination $d$. For each node $v$,
such a flow is defined as an indefinite sequence of packets with source $v$ and destination $d$. The goal is to minimize the \emph{load} of any link (or node) in the network,
which is defined as the number of flows crossing this link (or node). In case such a flow hits a link multiple times
the load of this edge is increased by $1$ for each such hit. An example of the outcome of all-to-one routing in the $K_5$ graph is given in \cref{fig:all-to-one}.
\begin{figure}
    \centering
    \includegraphics{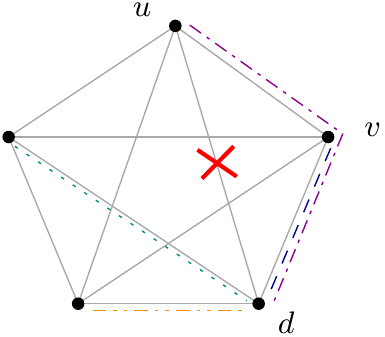}
    \caption{All-to-one routing in the complete graph $K_5$. Each node sends one flow towards destination $d$, each corresponding to one of the colored lines in non-solid style. Because the link $(u,d)$ is failed the flow of $u$ needs to take a detour. This causes the edge $(v,d)$ to accumulate a load of 2.}
    \label{fig:all-to-one}
\end{figure}
While, w.h.p., our protocols avoid packets to visit a node more than once, it may happen that some flows travel in a (temporary) forwarding loops for a small amount of hops. 
We ensure that packets of a flow are always forwarded along the same path, hence avoiding packet reorderings.

%% file: destination.tex
\section{Beating Deterministic Approaches with Three Permutations}

\label{sec:destination}	

This section presents our first failover algorithm.
While it is simple as forwarding is only based on the destination and hopcount header fields, it ensures w.h.p.~very low loads
even under a large number of link failures.

From the point of view of some fixed node $v$ the first protocol, we call it \emph{3-Permutations}, works as follows. For destination $d$, the node $v$ stores three permutations $\pi_{v,d}^{(1)}$, $\pi_{v,d}^{(2)}$ and $\pi_{v,d}^{(3)}$ of all nodes $u \in V \setminus \{d\}$. Each node chooses these three permutations uniformly at random. Upon receiving a packet $p$ intended for $d$, the node $v$ 
first tries to forward it directly via the link $(v,d)$. In case this link failed, $v$ inspects the current hop counter of $p$, denoted by $h(p)$. Depending on $h(p)$, the node $v$ then chooses one of the three permutations $\pi_{v,d}^{(i)}$ and forwards the packet to the first \emph{reachable} partner $w$ in this permutation. We call a node $w$  \emph{reachable} from $v$, if the direct link $(v,w)$ is not failed.  The criteria for selecting which permutation to use is simple. In case $h(p) < C_1$ for a value $C_1 = \Theta(\log n)$, permutation $\pi_{v,d}^{(1)}$ is consulted. For $C_1 \leq h(p) < 2 C_1$ the permutation $\pi_{v,d}^{(2)}$ is used and in any remaining case $\pi_{v,d}^{(3)}$ is utilized. In any case the packets hop counter is increased by $1$ before handing it to the next node. A concise description is given in \cref{alg:dest}.
	
Our main contribution is related to the way how these permutations are selected. Instead of opting for a deterministic protocol, we assume that each node $v$ chooses the permutations $\pi_{v,d}^{(i)}$ out of all possible permutations of nodes $u \in V \setminus \{d\}$ uniformly and at random. As the adversary is \emph{oblivious}, these permutations are not known to it and it needs to essentially blindly select edges for manipulation. Note however that this approach comes with a challenge. This random creation of failover routes may introduce temporary cycles into the packets routing paths. However, most of the packets $p$ reach the destination $d$ solely relaying on the failover entries given by the first permutation, $\pi_{v,d}^{(1)}$. And, only in case $p$ ends up trapped in a cycle, further permutations are used to allow it to escape said cycle.  We show that w.h.p., at most $O(\log^2 n \log \log n)$ load is accumulated at any node, even if the adversary is allowed to destroy a linear amount of edges.
	
\begin{figure}
\begin{algorithmic}[1]
\renewcommand{\algorithmicrequire}{\textbf{Input:}}
	\Require A packet with destination $d$ and hop count $h(p)$
	\If {$(v,d)$ is intact} forward $p$ to $d$ and \algorithmicreturn 
	\Else{ set $i$ to $\argmax_{j \in \{1,2,3\}} \{ h(p) \geq (j-1) C_1 \}$ and send $p$ to first directly reachable node in $\pi_{v,d}^{(i)}$}
	\EndIf
	\State $h(p)\mathrel = h(p) + 1$
\end{algorithmic}
\caption{\emph{3-Permutations} protocol. Point-of-view of some node $v$}
\label{alg:dest}
\end{figure}

	\begin{theorem}
	\label{thm:destination}
		Assume that the adversary fails  at most $\alpha \cdot n$ edges where $\alpha <1$ is a non-negative constant\footnote{More specifically, $\alpha$ can be an arbitrary constant with $0 < \alpha < (n-1) / n$. Note that this upper-bound quickly tends towards $1$ for large $n$.}. Then, if all nodes perform all-to-one routing to any destination $d$ and follow the \emph{3-Permutations} protocol, a maximum of 
		\[
			O( \log n \cdot \log \log n)
		\]
		flows passes at all but $O(\log^2 n)$ nodes. Furthermore, all remaining nodes, except for $d$, receive a load of at most $O(\log^2 n \cdot \log \log n)$ and every packet travels $O(\log n)$ hops. These statements hold w.h.p.
	\end{theorem}	
	
In order for the nodes to follow this protocol they require the exact value of $C_1$. This value upper-bounds the number of hops needed by any packet to reach the destination $d$, unless it is trapped in a cycle due to the permutation $\pi_{v,d}^{(1)}$. We show in \cref{lem:dest-with-inner-edges} that $C_1$ can be bounded from  above by $16 \log_{(1/\alpha)} n$. If $\alpha$ is not known to the nodes, then $C_1$ can be set to some value in $\omega(\log n)$. This slightly changes the result of \cref{thm:destination} where up to $O(\log^2 n)$ many nodes receive a load of $O(C_1 \cdot \log n \log \log n)$.

The reason why we employ exactly three permutations per destination is as follows. Either a packet ends up in a forwarding loop when forwarded via $\pi_{v,d}^{(1)}$ or it will reach the destination within $C_1$ hops. The same is true for packets being forwarded via $\pi_{v,d}^{(2)}$ and $\pi_{v,d}^{(3)}$. For each such permutation, the probability that the packet ends in a forwarding loop is $\polylog n / n$. Hence, only if the packet ends up in a forwarding loop for all three permutations, it will not reach the destination. The probability for this is $\polylog n / n^3$ as the permutations are generated randomly. Because packets following our protocol take a different path depending on their destination and source node (roughly $n^2$ possibilities), we need to multiply this probability of a bad event by $n^2$. Hence, the probability that \emph{any} packet does not reach its desired destination is $\polylog / n$, which is a low probability event.
When it comes to memory complexity, each node may store $3$ permutations of $n$ nodes for every destination $d$. A naive approach would therefore require
routing tables of size $O(n^2 \log n)$ bits to be prepared for routing to any arbitrary destination $d$.
This can be overcome as follows. First, each node $v$ only computes 3 random permutations $\pi_v^{(i)}$, $i=1,2,3$ on \emph{all} nodes. The permutation $\pi_{v,d}^{(i)}$ for each $d \in V$ is simply $\pi_{v}^{(i)}$, and thus we apply $\pi_{v}^{(i)}$ to obtain our failover strategy regardless of $d$. Note that if the edge $(v,d)$ is not failed, then any packet with target $d$ that reaches $v$ is sent directly from $v$ to $d$. If, however, $(v,d)$ is failed, then such a packet is sent to the first node $w$ in $\pi_{v}^{(i)}$ for which $(v,w)$ is not failed. \label{par:memory} Secondly, 
note that the nodes only consult their permutations up until the first reachable node (see Line 2 of \cref{alg:dest}).
Even if all $\alpha n$ failed edges are incident to the same node $v$, then at least one of the first $3 \log_{(1/\alpha)} n$ nodes in each of the permutations $\pi_{v,d}^{(i)}$ is  directly reachable from $v$ w.h.p. This follows from the fact that the adversary does not know the random bits generated at some node.
Therefore, nodes may truncate the permutations, storing only the first $3\log_{(1/\alpha)} n$ entries of each permutation. In the low probability event that none of the first $3 \log_{(1/\alpha)} n$ is directly reachable from $v$ (due to the failed edges), another $3 \log_{(1/\alpha)} n$ nodes are selected uniformly at random -- without replacement.
	Employing these improvements yields an improved total memory complexity of $O(\log^2 n / \log (1/\alpha)) = O(\log^2 n )$ per node.	
	
In the following, we consider some arbitrary but fixed node $d$ as destination. As we  establish probabilistic guarantees of at least $1 - n^{-(1 + \Omega(1))}$ for the statements in \cref{thm:destination} w.r.t. this fixed node, applying the union bound then implies that the results indeed hold for arbitrary destinations $d$ w.h.p.	

Regarding the tightness of our result, assume all $\alpha \cdot n$ failed edges are so called \emph{destination edges}, i.e., edges incident to the destination. Then, the flow starting at the other end $v$ of such a (failed) edge is first sent to a node selected uniformly at random from the set $V \setminus \{d\}$. The resulting distribution of the load can be seen as the outcome of throwing $\alpha n$ balls into $n-1$ bins \cite{RS98}, and the maximum load immediately reaches $\Omega(\log n / \log \log n)$ at some node w.h.p.

\subsection{Notation and Conventions}
\label{sec:destination-notation}

	 As we consider a fixed destination $d$ we omit it from the indices of our previously defined notation. Additionally, for $i \in \{1,2,3\}$ and some integer $1 \leq j \leq n-1$, we denote by  $\pi_v^{(i)}(j)$ the  $j$-th node in $v$'s  $i$-th permutation.

	\begin{definition}[Inner/Destination Edges and Good/Bad Nodes]
		We call each edge $(v,d)$  with $v \in V$ a \emph{destination edge} as it is incident to the destination $d$. All remaining edges are called \emph{inner edges}. Furthermore, we call a node $v \in V \setminus \{d\}$ \emph{good} if the destination edge $(v,d)$ is not failed. Otherwise, we call $v$ a \emph{bad} node. By $V_G$ and $V_B$ we denote the sets of good and bad nodes, respectively.
	\end{definition}
	
	The intuition behind calling such a node \emph{good} is that it may directly forward packets to $d$ when following the protocol in \cref{alg:dest}.
Additionally, we define $\mathcal{F}$ to be the set of failed edges and further partition this set into $\Fin$ and $\Fout$. The former contains all failed inner edges, the latter the failed destination edges. We let $\varepsilon$ and $\gamma$ be constants such that  $|\Fout| \leq \varepsilon \cdot n$ and $ |\Fin| \leq \gamma \cdot n$, where $\varepsilon + \gamma \leq \alpha < 1$. 

	If we state that we apply Chernoff bounds for a random variable $X$, we mean the multiplicative variant $P[X \geq (1+\delta)\mu ] \leq \exp{(\min\{\delta,\delta^2\} \mu / 3)}$, where $\mu=E[X]$ and $\delta > 0$. For lower tails we use $P[X \leq (1 -\delta) \mu ] \leq \exp{( \delta^2 \mu /2)}$ for $0 < \delta < 1$ (see e.g. \cite{Drr2020}). Similar, if we say we apply \emph{union bounds}, we mean Boole's inequality  \cite{Drr2020}. For a set of probabilistic events, this bound states that the probability of at least one event in this set occurring is no greater than the sum of the probabilities of the individual events.
Besides w.h.p., we introduce the following abbreviations: Instead of \emph{with probability}, \emph{uniformly at random} and \emph{random variable}, we use w.p., u.a.r.~and r.v., respectively. We denote the binomial distribution with $m$ trials and success probability $p$ by $B(m,p)$. Finally, by $\log n$ we denote $\log_2 n$.	
	
\subsection{Analysis}
\label{sec:destination-analysis}

 We first establish some structural properties that describe the paths that the packets take according to our failover strategy. For $i \in \{1,2,3\}$, we define the directed subgraphs $G^{'(i)} = (V \setminus \{d\}, E^{'(i)})$ with edge sets $E^{'(i)} = \{ (v, \pi_v^{(i)}(1)) ~|~ v \in V_B \}$. Under assumption that $\Fin = \emptyset$, this graph depicts the possible paths a packet traverses to either the good nodes $V_G$ or to some possibly existing cycle. It is easy to see that any graph $G^{'(i)}$ hosts multiple trees, each rooted in some $w \in V_G$ as the nodes in $V_G$ are the only ones with out-degree $0$. 
The other components in $G^{'(i)}$ do not contain a node of $V_G$. Instead, they contain a cycle, in which each node on the cycle\footnote{Such a cycle may consists of  a single node $v \in V_B$ if $E^{'(i)}$ contains the edge $(v,v)$. This happens in case of $\pi_v^{(i)}(1) = v$.} is the root of a subtree (see component on the right of \cref{fig:forest}).
The first important result of our analysis is that the size of these structures is at most $O(\log n \cdot \log \log n)$ w.h.p.~

In the next step, we account for the failures in $\Fin$. Here we use the fact the permutations are chosen completely at random. We show that only $O(\log n)$ of all edges in the graphs $G^{'(i)}$ are failed, which reinforces the intuition that failing inner edges has little effect compared to the failure set $\Fout$. This approach allows us to construct the graphs $G^{''(i)}$, which now account for inner edge failures and correctly depict the paths that the packets traverse.
	
	Finally, we put everything together and use the  graphs $G^{''(i)}$ to show the result of \cref{thm:destination}. At this point we also argue that $3$ permutations per node do indeed suffice for \emph{any} packet to be routed to $d$ w.h.p.~
	
\paragraph{Measuring Forests}	
As mentioned we start by analyzing the graphs $G^{'(i)}$. We first consider some fixed $i \in \{1,2,3\}$ and omit the superscript $(i)$. That is, we consider the graph $G'$ together with the edge set $E' = \{ (v, \pi_v(1)) | v \in V_B\}$ and $V' = V \setminus \{d\}$.
 As already discussed, only the existence of some cycles between nodes in $V_B$ prevents $G'$ from being a forest. A rough perspective on $G'$ is given in \cref{fig:forest}.

\begin{figure}
\centering
\includegraphics[scale=0.85]{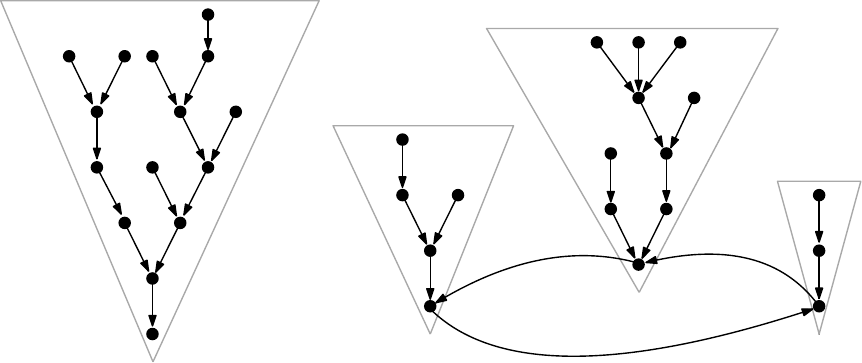}
\caption{The structures contained in the subgraph $G'$. On the left, a tree rooted in some $v \in V_G$ is presented. On the right, we have a cycle and each node of the cycle is again the root root of a tree.}
\label{fig:forest}
\end{figure}

We start by establishing some structural properties of the graph $G'$.
	
	\begin{lemma}
	\label{lem:dest-path-length}
		The graph $G'$ does not contain paths or cycles of length larger than $4 \log_{1/\varepsilon} n$ w.p.\ $1-n^{-3}$. Additionally, the number of  cycles in $G'$ is $O(\log n)$ w.p.~$1-n^{-3}$.
	\end{lemma}

	\begin{proof}
		Let $v \in V_B$ be an arbitrary node with $(v,d) \in \Fout$ and consider the edges $(u, \pi_u(1))$ of all $(1- \varepsilon)n$ nodes $u \in V_B$. 
Starting at $v$ we uncover the outgoing edges one after the other and follow the resulting path until either a node $w \in V_G$ is reached or a cycle is created.		
This way, a path of at least $i+1$ distinct nodes is traversed w.p.~at most 
		\[
			\frac{\varepsilon n - 1}{n -1} \cdot \frac{\varepsilon n -2}{n-1} \cdot 
			... \cdot \frac{\varepsilon n - i}{n-1} < \varepsilon^i \left( 1 + o(1) \right) < \varepsilon^{i+1}.
		\] 
		Here each quotient corresponds to the probability of continuing the path for one further step without hitting a node $u \in V_B$ that we already visited.  Therefore, our fixed node $v$ is at the beginning of a path of length at least $4 \log_{1/ \varepsilon} n$ with probability at most $n^{-4}$.\\
        To bound the probability that $\emph{none}$ of the nodes in $V_B$ are the origin of such a long path we apply union bounds. As $|V_B|$ such nodes exist, and each of these nodes has probability $n^{-4}$ to form such a path, we conclude that this probability is at most $|V_B| \cdot n^{-4} <n \cdot n^{-4} = n^{-3}$. Note that this also upper bounds the length of the cycles, as one may see a cycle as a path that is concluded with a previously visited node. 
		
		To determine the number of cycles contained in $G'$, consider again some node $v \in V_B$ and the path that is created by the process described above.  Assuming that we are at the last node before the path terminates, there are two possible outcomes: Selecting one of the w.h.p.~at most  $4 \log_{1/\varepsilon} n$ already visited nodes and creating a cycle, or one of the  $(1-\varepsilon) n$ nodes $v \in V_G$. Therefore, a cycle is created by $v$ w.p.~at most
		\[
			\frac{4 \log_{1/ \varepsilon} n}{ (1- \varepsilon) n + 4 \log_{1 / \varepsilon } n}.
		\]

	Now we sequentialize this process for all $v \in V_B$ one after the other. Clearly there exist dependencies as it is possible that $v \in V_B$ is already contained in a path that was already uncovered before. Similarly, the path starting in $v$ might hit the path of another node, which was uncovered earlier. Note however that this only decreases the probability for $v$ to create a new cycle. Therefore w.p.~$p = O(\log n/ n)$ some fixed node $v \in V_B$ creates a new cycle regardless of the other nodes.
	Chernoff bounds immediately yield the desired result.
\end{proof}

Consider again the graph $G' = (V',E')$ with $V' = V \setminus \{ d \}$ and some fixed node $w \in V_G$. Remember that such a node is the root of a tree in $G'$, induced by the edges $(v,\pi_v (1))$ for $v \in V_B$.
Let now $L_i$ denote the set of nodes at level $i$ of $w$'s tree, where $L_0 = \{ w \}$. We now construct $L_1,L_2, ...$ step by step as follows. In the $i$-th step construct $L_i = \{v ~|~ v \in (V_B \setminus \bigcup_{j=0}^{i-1} L_j) \wedge \pi_v(1) \in L_{i-1}\}$. One can see this as constructing the tree $w \in V_G$ layer-by-layer, by adding nodes with outgoing edges connected to nodes in the set $L_i$ in the $i$-th step.
We define the r.v.~$X_i = |L_i|$ and when fixing $w \in V_G$ we  call $\{X_i\}$ the \emph{layer sequence}, or in short \emph{sequence}, corresponding to $v$. The step-wise construction described above directly yields the following lemma.

\begin{lemma}
\label{lem:dest-level-binomial}
	Fix some root node $w \in V_G$ in $G'$ together with its layer sequence $\{ X_i\}$. Then, it holds for level $i$ that $X_i \sim B(m,p)$ with
		\begin{align*}
			m = |V_B| - \sum_{j=1}^{i-1} X_j \text{ and } p = \frac{X_i}{|V'| - \sum_{j=0}^{i-2} X_j}
		\end{align*}
\end{lemma}

\begin{proof} 
	Let $w \in V_G$ and consider the process of uncovering edges just as described in the paragraph before this lemma. Clearly it holds that $X_1 \sim B(|V_B|, 1/V')$ as the partner of every $v \in V_B$ is chosen u.a.r.~due to the nature of the permutations $\pi_v$. 
	
	Assume now we are in the $(i-1)$-th step and already uncovered all edges between the sets $L_0,L_1, ... , L_{i-2}$. To construct $L_{i-1}$ we need to consider edges $(v, \pi_v(1))$, where $\pi_v(1) \in L_{i-2}$ and $v \in V_B \setminus \bigcup_{j=0}^{i-2} L_j$. At this point $m:=|V_B| - X_1 - X_2 - ... - X_{i-1}$ nodes still have uncovered outgoing edges. We know that these edges do not connect to nodes in the sets $L_0, ... , L_{i-2}$. Now, as the edge partners are chosen u.a.r the probability that such an edge connects to nodes in $L_{i-1}$ is exactly $p := X_{i-1} / (|V'| - X_0 - X_1 - ... - X_{i-2})$. Note that there exist no dependencies between the selections of these $m$ nodes. Hence, $X_i$ follows   $B(m,p)$.
\end{proof}

This means that we can describe the tree rooted in $w \in V_G$ by a sequence of binomial distributions, whose expected value depends on the previous layers.
The following statement gives us a bound on $m \cdot p$ w.h.p., showing that the set of nodes at level $i$ indeed decreases exponentially fast. The proof is given in \cref{appendix} and  mostly relies on the statement of \cref{lem:dest-level-binomial} in conjunction with Chernoff bound applications. 
\begin{lemma}
\label{lem:level-shrinking}
	Let $w \in V_G$ be a root node in $G'$ together with its corresponding layer sequence $\{X_i\}$. Then, there exists a constant $0 < \beta < 1$ such that for $i < \log^2 n$ it holds that
$
		E[X_{i+1}] \leq X_i \cdot \varepsilon (1 + o(1)) \leq X_i \cdot \beta 
$.
	Additionally, there exists a constant $C > 0$ such that, w.p.~$1-n^{-3}$ it holds for all $i \leq \log^2 n$ that $X_i < C \log n$. 
\end{lemma}

According to \cref{lem:dest-path-length} packets travel at most $O(\log n)$ hops until reaching the destination. This implies the following.

\begin{corollary}
\label{lem:tree-height}
Consider \emph{any} $w \in V_G$ with its corresponding layer sequence $\{X_i\}$ in $G'$. Then, for $i>C' \log n$ it holds that $X_i = 0$ w.p.~at least $1-n^{-3}$.
\end{corollary} 
	
Clearly for some fixed node $w \in V_G$ with sequence $\{X_i\}$, our main interest lies in the value $X = \sum_{i \geq 0} X_i$. In the following we say that $X_{i-1} < a$ \emph{increases} into the interval $[a,b)$, iff $X_i \in [a,b)$. Analogously we say $X_{i-1} \geq b$ \emph{decreases} into the same interval iff $X_{i} \in [a,b)$.
	
	\begin{lemma}
	\label{lem:tree-interval}
		Consider again a root $w \in V_G$ and the corresponding  layer sequence $\{X_i\}$. Then, for $j< \log ( \log n / \log \log n)$ and a constant $\hat{\beta}$ with $\beta < \hat{\beta} < 1$ the following holds: at most $O(\hat{\beta}^{-j})$ members of $\{X_i\}$ increase into the interval
		\[
			\left[ C \log n \cdot \hat{\beta}^{j}, C \log n \cdot \hat{\beta}^{j-1} \right) 
		\]		
		w.p.~at least $1 - n^{-3}$. Note that $\beta$ and $C$ are the constants defined in \cref{lem:level-shrinking}.
	\end{lemma}

	\begin{proof}
		In the following we consider the elements of the sequence $\{X_i\}$ one after the other, starting with $X_0$. By \cref{lem:dest-level-binomial} and \cref{lem:level-shrinking} we know that the $i$-th value $X_i$ follows a binomial distribution with mean less than $X_{i-1} \cdot \beta$. Using Chernoff bounds together with the fact that $\beta < 1$ is a constant, we obtain for any $t \geq 0$ 
\begin{equation}
\label{eq:dest-tree-interval}
	Pr[X_i \geq t~|~X_{i-1} \leq t] \leq \exp(-\Omega(t)).
\end{equation}
 Note that for $t= C\log n \cdot \hat{\beta}^{j}$ ,  (\ref{eq:dest-tree-interval}) bounds the probability that  $X_{i-1}$ increases into $[C\log n \cdot \hat{\beta}^{j}, C\log n \cdot \hat{\beta}^{j-1} )$.  Now, from \cref{lem:tree-height} it follows that at most $C' \log n$ elements may increase into the interval mentioned before. Therefore we can majorize the total number of increases into the interval by $B(C' \log n, \exp ( -c \cdot t))$, where $c$ is the constant hidden in $\Omega(t)$ in (\ref{eq:dest-tree-interval}). Now, using the well-known upper bound $\binom{s}{t} \leq \left(\frac{e \cdot s}{t}\right)^t$ on the binomial coefficient (see e.g. Proposition B.2 of \cite{motwani-book}) we get
\begin{align}
\label{eq-al:dest-tree-interval}
	Pr \left[ B \left( C' \log n, \exp \left( -c C \cdot \log n  \cdot \hat{\beta}^j \right) \right) = \frac{5}{cC} \cdot \hat{\beta}^{-j} \right] < \\ \left( O(1) \log n \cdot \hat{\beta}^j \right)^{O \left(\hat{\beta}^{-j} \right)} \left( \frac{1}{e} \right)^{5 \log n} 
	< \frac{1}{n^{3.4}} \nonumber 
\end{align}
for $n$ large enough and $j < \log ( \log n / \log \log n)$.
\end{proof}	

We are finally ready to state that no tree contained in $G'$ consists of more than $O(\log n \cdot \log \log n)$ nodes.

\begin{lemma}
\label{lem:tree-bound}
	Let $w \in V_G$ be a root and $\{ X_i\}$ the corresponding layer sequence. Then it holds w.p.~at least $1 - \polylog n / n^{3}$ that $\sum_{i} X_i < O(\log n \cdot \log \log n)$
\end{lemma}
\begin{proof}
	We start by fixing some interval $ J:= [C\log n \cdot \hat{\beta}^j, C\log n \cdot \hat{\beta}^{j-1})$ for $j < \log ( \log n/ \log \log n)$ just as in \cref{lem:tree-interval}. In the following we consider the so-called \emph{cost} caused by this interval, i.e, $\sum_{\{i:~X_i \in J\}} X_i$. Three possible events may cause some $X_i$ to contribute to this sum:
	\begin{enumerate}
		\item $X_{i-1}$ \emph{increased} into the interval
		\item $X_{i-1}$ \emph{decreased} into the interval
		\item $X_{i-1} \in J$ and $X_i \in J$
	\end{enumerate}
Now, the first point is addressed by \cref{lem:tree-interval}. We know that w.h.p. at most $O(\hat{\beta}^{-j})$ such values exist in total. As for the second point,
assume that this event occurred for some $X_{i-1}$, i.e., assume  $X_{i-1}$ decreased into $J$. Clearly, for $X_{i'-1}$ with $i' > i$ to  decrease into $J$ again, there must be a $k \in {i, \dots ,  i'-2}$ such that $X_k$ increases into an interval with index $j' > j$.
Using again \cref{lem:tree-interval} and applying the union bound over all intervals with $j < \log ( \log n / \log \log n)$ we know that at most $O \left( \hat{\beta}^{-(j-1)} + \hat{\beta}^{-(j-2)} + ... + \hat{\beta}^{-1} \right) = O \left( \hat{\beta}^{-j} \right)$ such members exist.

The third point is the most interesting. Assume $X_{i-1} \in [C\log n \cdot  \hat{\beta}^j, C\log n \cdot \hat{\beta}^{j-1})$, and observe that if $X_{i} \leq \hat{\beta} X_{i-1}$ it follows that $X_i \not \in J$.
Now, according to \cref{lem:level-shrinking} we know that $E[X_i] \leq X_{i-1} \beta$. Noting that $\hat{\beta} > \beta$ we bound $Pr[X_i \geq X_{i-1}\cdot \hat{\beta}]$. Applying Chernoff bounds with $\delta =  \hat{\beta} / \beta - 1 = \Omega(1)$, we get
\[
	Pr[X_i \geq X_{i-1} \hat{\beta}] < \exp (-\Omega(X_{i-1})) < \exp \left(-\Omega(\log n \cdot  \hat{\beta}^{j})\right).
\]
As the whole sequence $\{X_i\}$ has length at most $C' \log n$ according to \cref{lem:tree-height}, any of the above events can happen at most $C'\log n$ times. A similar approach as in (\ref{eq-al:dest-tree-interval}) therefore immediately yields that at most $O(\hat{\beta}^{-j})$ many times $X_i$ stays in the interval $J$ w.p.~at least $1 -n^{-3}$.

Summarizing, the cost caused by elements taking values in $J$ is at most $ O(\hat{\beta}^{-j} \cdot C \log n \cdot \hat{\beta}^{j-1} ) = O( \log n)$ w.p.~$1 - n^{-3}$. When applying the union bound we obtain that the total cost generated by all intervals with 
$j \leq \log ( \log n / \log \log n)$ is $O(\log n \cdot \log \log n)$. Note that for $j > \log ( \log n/ \log \log n)$ it holds that $C \log  n \cdot \hat{\beta}^{j-1} = O( \log \log n)$. Even if the whole sequence $\{X_i\}$ remains in these remaining intervals for all $C' \log n$ steps, a cost of at most $O(\log n \cdot \log \log n)$ can be accumulated.
\end{proof}

When applying the union bound, this gives us that \emph{no} tree with root $w \in V_G$ of $G'$ exceeds size $O(\log n \cdot \log \log n)$ w.h.p.~However, remember that another type of component exists in $G'$, namely cycles that may have additional nodes attached to them. Fix one of these components and let $C = \{v_1, ... ,v_{|C|}\}$ be the set of nodes of the cycle. Furthermore define $A_i := \{v ~|~ v \in (V_B \setminus C) \wedge  \text{ a path from $v$ to $v_i$ exists in } G'\} \cup \{v_i\}$. Then, this set induces a tree in $G'$ unless the loop $(v_i,v_i)$ is contained in $E'$ which implies $|C| = 1$. In any case, one may see the component as being induced by the set $\bigcup_i A_i$ as a forest of trees, whose roots lie on the cycle $C$. To determine the total size of the set $\bigcup_{i=1}^{|C|} A_i$, we look at the growth of them layer-by-layer. However, this time we let all $|C|$ of these trees grow at the same time, again uncovering edges step-by-step. That is, $L_0 = C$ and construct any set $L_i$ with $X_i =|L_i|$ just as we did when considering the roots $w \in V_G$, i.e., 
$L_i=\{ v ~|~ v \in (V_B \setminus \bigcup_{j=0}^{i-1} L_j) \wedge \pi_v(1) \in L_{i-1}\}$.
Observe that for this fixed cycle we can describe the sequence $\{X_i\}$ by the number of nodes in $V_B \setminus C$ being at distance $i$ from the cycle. When replacing the set $V_B$ by $V_B \setminus C$, the result of \cref{lem:dest-level-binomial} also holds for this layer sequence. As \cref{lem:dest-path-length} guarantees that $X_0 < O(\log n)$, the statement of \cref{lem:level-shrinking} follows accordingly and allows us to repeat the whole approach.

\begin{corollary}
\label{lem:tree-bound-cycle}
	The results of \cref{lem:level-shrinking}, \cref{lem:tree-height}, \cref{lem:tree-interval} and finally \cref{lem:tree-bound} also hold for the sequence $\{X_i\}$ of nodes in distance $i$ to some fixed cycle in $G'$.
\end{corollary} 
	
	\paragraph{Accounting for Inner Edge Failures} 
	So far we established that none of the components in $G'$ is of size more than $O(\log n \cdot \log \log n)$ w.h.p.~Remember however that we still need to account for the failures in $\Fin$. We start by showing that only $O(\log n)$ of them lie on any path in $G'$, even if the adversary fails $\Theta(n)$ inner edges.

	\begin{lemma}
	\label{lem:dest-log-inner-failures}
		The number of nodes $v \in V$ that have their first failover edge $(v, \pi_v(1))$ destroyed by the adversary is $O( \log n)$ w.p.~$1-n^{-3}$
	\end{lemma}

	\begin{proof}
		Consider some fixed node $v$ and assume that the adversary destroyed $f_v < n-1$ inner edges connected to~$v$. Remember, the adversary cannot predict the permutation. Therefore a failed edge lies at the beginning of $\pi_v$ w.p.~at most $f_v / (n-1)$. Define r.v. $X_v$ s.t. $X_v=1$ iff $(v,\pi_v(1)) \in \Fin$, and $0$ otherwise. Clearly the random variable $X = \sum_{v \in V} X_v$ is what we are looking for where
		\[
			E[X] = \sum_{v \in V \setminus \{d\}} X_v = \frac{1}{n-1} \cdot \sum_{v \in V \setminus \{d\}} f_v \leq \frac{2\gamma  n}{n-1} < 2.
		\]
		In the third step, we use that $\sum_{v \in V \setminus \{d\}} f_v = 2 \cdot |\Fin| < 2\gamma n$. The factor $2$ is introduced because the sum represents the number of incident failed edges, summed up over all nodes. This way, each of the $|\Fin| < \gamma n$ failed inner edges is counted twice.
		As all $X_v \in \{0,1\}$ are independent from each other, we may apply Chernoff bounds to $X$. For $\delta = 9 \ln n/ E[X] < \delta^2$ this yields that $Pr[X > E[X] (1 + 9\ln / E[X])] \leq \exp(-3 \ln n) = n^{-3}$. Hence $\Pr[X > 2 + 9\ln] \leq n^{-3}$.
	\end{proof}

In the following we show how to transfer $G^{'(i)}$ into $G^{''(i)}$ for any $i =1,2,3$. The basic idea is to remove edges $(v, \pi^{(i)}_v(1))$ that lie in $\Fin$ from $G^{'(i)}$ and replace them with edges $(v, \pi^{(i)}_v(j))$, where $\pi^{(i)}_v(j)$ is the first reachable neighbor in the permutation of $v$. This way, the graph $G^{''(i)}$ represents the correct path of the packets with hop counter $(i-1)C_1 \leq h(p) < i C_1$ when following our protocol. The edge replacements throughout the construction of  $G^{''(i)}$ causes subtrees of size $O(\log n \cdot \log \log n)$ to relocate (see \cref{fig:tree_cutout}). This may cause some components in $G^{''(i)}$ to be extended by these relocated subtrees, and also a new type of component may be created. That is, the roots of relocated subtrees may connect with nodes in other subtrees and cause the formation of a cycle.

Formally, we can show the following claim.
As in the previous sections, we focus on a fixed $i$ and will omit the superscript $(i)$.  The idea behind the proof is that, according \cref{lem:dest-log-inner-failures}, only $O(\log n)$ such subtrees are relocated. It is unlikely that more that $O(1)$ of these subtrees connect to the same structure. Full proof is given in \cref{appendix}.

	\begin{lemma}
\label{lem:dest-with-inner-edges}
	Consider the graph $G''$. Then, none of the components contained in $G''$ has more than $O(\log n \cdot \log \log n)$ nodes. Furthermore, the number of contained cycles remains in $O(\log n)$ with each not exceeding length $O(\log n)$. Additionally, any packet that is not trapped in some cycle takes at most $C_1 = 16 \log_{1 / \varepsilon } n$ steps to reach $d$. All above statements hold w.p.~$1-O(n^{-2})$
\end{lemma}
	
\begin{figure}
\centering
\includegraphics[scale=0.8]{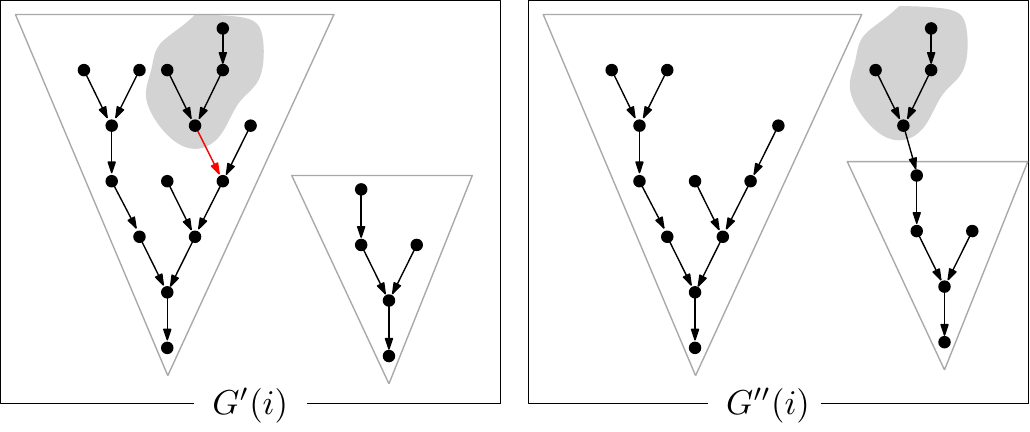}
\caption{For some node $v$ the edge $(v, \pi_v^{(i)} (1))$ is failed (marked in red). In $G''(i)$ this edge is replaced by $(v, \pi_v^{(j)})$, causing the subtree rooted in $v$ to relocate. }
\label{fig:tree_cutout}
\end{figure}

\paragraph{From Forests to Load}
To determine the total load some node $v$ receives we look at the graphs $G^{''(i)}$ one after another. When starting the \emph{all-to-one} routing, each node initiates a flow and all of them follow  paths according to the single outgoing edges in $G^{''(1)}$. Now, most of the flows reach the destination $d$ after at most $C_1$ many hops. However, some might be trapped inside a cycle. In both cases consider the cost $T_1(v)$ that occurred at any node $v$ until this point. Upon reaching a hop value of $C_1$ the flows currently trapped in a cycle start to traverse according to the permutation $\pi_v^{(2)}$. If we know where these cycle nodes are located in $G^{''(2)}$, we can again track the flows' paths and determine the loads caused by the next $C_1$ many hops, denoted $T_2(v)$. Finally, we repeat this approach one more time in the same manner and obtain the values $T_3(v)$. In the following we say that flows \emph{exit} the system $G^{''(i)}$ if their respective hop counter reaches $i C_1$. Similarly we say flows \emph{enter} the system $G^{''(i)}$, if they reach hop value $(i-1)C_1$ and $G^{''(i)}$ becomes relevant for their failover paths for the first time. Finally, we argue that $T_i(v) = 0$ w.h.p.~for any $v$ and $i \geq 4$, which is why we originally only required three permutations, and deduce that the total load $v$ receives is $T(v) = T_1(v) + T_2(v) + T_3(v)$.

We start with the following statement. The proof, given in \cref{appendix}, exploits the structure of $G^{''(i)}$, as each node has out-degree either one or zero.
\begin{lemma}
\label{lem:dest-cycle-load}
	Let $i \geq 1$ and assume that every component in $G^{''(i)}$ is entered by $O(\log n \cdot \log \log n)$ total flows.
	Then, every node $v$ that is not contained in a cycle receives a load of at most $T_i(v) = O(\log n \cdot \log \log n)$. Nodes contained in a cycle in $G^{''(i)}$ receive $T_i(v) = O(\log^2 n  \cdot \log \log n)$ load until the flows exit $G^{''(i)}$ w.p.~$1-O(n^{-2})$
\end{lemma}

Clearly, for $i=1$ and $G^{''(1)}$ the assumption of above lemma is satisfied. This is  because \cref{lem:dest-with-inner-edges} guarantees that all structures are of size $O(\log  n \cdot  \log \log n)$ and each node starts sending one flow of packets. The next lemma deals with the inductive step and $i>1$. The proof relies on the fact that  the permutations $\pi_v^{(i+1)}$, $v \in V'$, were chosen independently from any previous permutation $\pi^{(j)}_v$, $j < (i+1)$.  A detailed proof is given in \cref{appendix}.
\begin{lemma}
	\label{lem:dest-cycle-load-assumption}
	Assume that the assumption of \cref{lem:dest-cycle-load} holds for $G^{''(i)}$. Then, also in $G^{''(i+1)}$ no component is entered by more than $O(\log n \log \log n)$ flows w.p.~$1-O(n^{-2})$
\end{lemma}

Above lemma concludes the inductive step, showing that \cref{lem:dest-cycle-load} is applicable to all $i \geq 1$.  We now state that it is indeed enough to only look at the graphs $G^{''(i)}$ where $1 \leq i \leq 3$ to determine the maximum load. The proof of the statement is very  similar to the proof of \cref{lem:dest-path-length} and given in \cref{appendix}. We track a fixed packet that starts at some node and argue that it is unlikely to be stuck in a cycle in  $G^{''(1)}, G^{''(2)}$ and $G^{''(3)}$.

\begin{lemma}
\label{lem:g3-enough}
	Fix an arbitrary packet $p$ sent by a node $v \in V_B$. Then, $p$ ends up in a cycle in $G^{''(1)}, G^{''(2)}$ and $G^{''(3)}$ w.p.~at most $\polylog n / n^{3}$. Additionally, it holds that $T_i(v) = 0$ for all nodes $v  \in V \setminus \{d\}$ and $i \geq 4$ w.p.~$1- \polylog / n^{2}$.
\end{lemma}

Note that this lemma also implies that $O(\log n)$ hops suffice for any packet to reach the destination. We established in \cref{lem:dest-cycle-load-assumption} that the assumption in \cref{lem:dest-cycle-load} is indeed fulfilled for $G^{''(1)}, G^{''(2)}$ and $G^{''(3)}$. Hence, the total load any node receives that is not contained in a cycle in either $G^{''(2)}$ or $G^{''(1)}$ is $O(\log n  \cdot \log \log n)$. And those that lie on a cycle have a load of $O(\log^2 n  \cdot \log \log n)$, where according to \cref{lem:dest-with-inner-edges} at most $O(\log^2 n)$ such nodes exist.
\cref{thm:destination} follows accordingly.

%% file: interval.tex
\section{Circumventing Cycles by Partitioning}
\label{sec:interval}

	One of the major challenges of the \emph{3-Permutations} protocol is to cope with temporary cycles, which may be introduced when employing randomization into the failover strategy. For packets with large hop counts, which indicate the existence of a cycle, we effectively provided different failover routes.
In the following we present the \emph{Intervals} routing protocol that \begin{enumerate*} \item does not introduce any cycles in the packets routing paths w.h.p.~and \item is purely destination based \end{enumerate*}. This comes however at the cost of smaller maximum resilience against failures.
	
	The \emph{Intervals} protocol works as follows. We assume that every node is given a unique ID or address, which is known to the other nodes. Therefore we can enumerate the nodes by $v_1, v_2, ... ,v_n$. Let now $\alpha$ be some small constant $0 < \alpha < 1$. 
	We partition the nodes of the graph into consecutive sets of size $ I = n/4 \log_{1/\alpha} n$. That means, the $i$-th set $R_i$ contains nodes with addresses in the range 
	\[ \Big [(i-1) \cdot \frac{n}{4 \log_{1/\alpha} n} + 1  ~,~ i \cdot \frac{n}{4 \log_{1/\alpha} n} \Big] .
	\]
	 Assume that the value $\alpha$ is chosen such that both, the interval bounds and the number of intervals, are integers. The next step is similar to \cref{sec:destination}.  Every node tries to directly forward a packet to the desired destination $d$, if the direct link is available. Otherwise, again a permutation $\pi_{v,d}$ of nodes is consulted and the packet is sent to the first reachable partner in $\pi_{v,d}$. The crucial difference to the \emph{3-Permutations} protocol is the following: for some node $v$ that lies in the interval $R_i$, the permutation $\pi_{v,d}$ is a permutation of the nodes in $R_{(i+1)} \setminus \{d\}$. Hence, only edges ranging from nodes in the set $R_i$ to nodes in $R_{i+1}$ are considered as possible failover edges. To allow for a proper protocol, we assume that nodes of the rightmost interval choose failover edges into $R_0$.	We show in the upcoming analysis that this protocol does not create any cycles in the routing paths w.h.p.~The following statement allows the adversary to fail up to $O(n / \log n)$ many edges. 
Note that this protocol is purely destination based and therefore any deterministic scheme operating under this constraint would allow  $\Omega(n /\log n)$ load to be created by the adversary \cite{opodis13shoot}. 

\begin{figure}
\begin{algorithmic}[1]
\renewcommand{\algorithmicrequire}{\textbf{Input:}}
	\Require A packet with destination $d$
	\If {$(v,d)$ is intact} forward $p$ to $d$ and \algorithmicreturn \EndIf
	\State Forward $p$ to first directly reachable node in $\pi_{v,d}$
\end{algorithmic}
\caption{\emph{Intervals} protocol. Point-of-view of some node $v$}
\end{figure}

\begin{theorem}
\label{thm:interval}	 
	Assume the adversary is allowed to fail up to $\alpha \cdot I$ many edges, for some arbitrary constant $0 <\alpha < 1$ where $I = n / (4 \log_{1/\alpha} n)$. Then, when considering all-to-one routing to any destination $d$, the \emph{Intervals} protocol guarantees a maximum of
	\[	
		O(\log n  \cdot \log \log n)
	\]
load at any node except $d$ and edge w.h.p.~Additionally, no packet performs more than $O(\log n)$ hops w.h.p.~
\end{theorem}

For $\alpha = 1/e$ above statement provides the maximum resilience of $n/(4e \log_{e} n$).
Furthermore, assume the adversary fails $\Omega(n/\log n)$ destination edges $(v,d)$, with all such nodes $v$ being in the same interval $R_i$. Then, similar as in the case of the \emph{3-Permutation} protocol, a  balls-into-bins argument \cite{RS98} shows that after all nodes in $R_i$ send their flows to  randomly chosen nodes in $R_{i+1}$, at least one node has load $\Omega(\log n / \log \log n)$ w.h.p. 

Concerning the notation we carry over everything defined in \cref{sec:destination-notation}. Additionally we extend the notation for the sets of failed edges as $\Fin_i$ and $\Fout_i$. These sets only contain failed edges started in the $i$-th interval and, in case of $\Fin_{i}$, have partners in the interval $i+1$. Just as in the protocols description we denote the set of nodes inside the $i$-th interval as $R_i$ where $I = |R_i|$. 

We remark that the result of \cref{thm:interval} also holds, if we require that $|\Fin_i| < \varepsilon \cdot I$ and $|\Fout_i| < \gamma \cdot I$ \emph{for any} interval $R_i$, as long as $\varepsilon + \gamma \leq 1 - \Delta$ for some constant $\Delta > 0$. As the sets $\Fin_i$ and $\Fout_i$ are specified for some fixed destination $d$, we require this property for all possible destinations $d$.

Regarding the memory complexity, for a fixed destination each node needs a permutation of $O(n / \log n)$ nodes, hence $O(n)$ bits in total. 
As in case of the \textit{3-Permutations} protocol, the set of permutations $\{\pi_{v,d} ~|~ d \in V\}$ some node $v$ requires can be derived from a single permutation $\pi_{v}$, and only the first $3 \log_{1/\alpha} n$ entries of each permutation need to be stored (see corresponding description on page \pageref{par:memory}). This allows for a reduction of the \emph{total} memory required per node to $O(\log^2 n / \log (1/\alpha) )$.

	\subsection{Analysis}
	
A main motivation behind \cref{sec:interval} is to eliminate the need to perform any kind of cycle resolution. To that end we start by showing that our protocol does not introduce any cycles into the packets routing paths and at the same time derive that the hop count of \emph{any} packet remains in $O(\log n)$.
		
	To show that the maximum load occurring lies in $O(\log n \cdot \log \log n)$, we take a similar approach as in \cref{sec:destination-analysis}. That is, we fix some node $v \in V_G \cap R_i$ in some interval. Then, all the sources of packets that are forwarded over the edge $(v,d)$ form a tree rooted in $v$. We again argue that, in expectation, the number of nodes per level of this tree decreases exponentially fast and reuse parts of \cref{sec:destination-analysis}.
	
	\paragraph{Cycles} Some packet located in $R_{i}$ has only two possibilities for its next hop. Either it is directly forwarded to the destination, or it is forwarded to a node of the set $R_{i+1}$. To end up in a cycle it needs to traverse a sequence of $4 \log_{1/\alpha} n$ intervals, hitting a node $v \in V_B$ with every hop. This is unlikely and formalized as follows.
	
	\begin{lemma}
	\label{lem:interval-paths}
		Let $p$ be an arbitrary packet to be routed to destination $d$. Then, its routing path does not contain any cycles and it reaches the destination within $4 \cdot \log_{1/\alpha} n +1$ hops w.p.~at least $1 - n^{-4}$.
	\end{lemma}
	\begin{proof}
		Consider some packet $p$ originating from some node $v \in V_B$ in an arbitrary interval $R_i$. In the next hop, $p$ travels to some node $v'$ in $R_{i+1}$. If $v' \in V_G$ the packet is then forwarded directly to $d$. 
	Now, $|\Fout_{i+1}|$ bad nodes exist in $R_{i+1}$. In the worst-case the adversary fails  $|\Fin_i|$ edges $(v,w)$ with $w \in R_{i+1} \cap V_G$. Hence, the probability that $v'$ is a bad node is at most
	\[ 
		\frac{|\Fout_{i+1}|}{ I - |\Fin_{i}| } < \alpha.
	\]
Here we used that $|\Fin_i| + |\Fout_{i+1}| \leq \alpha I$. Therefore, the probability that $p$ hits bad nodes in $4 \cdot \log_{1/\alpha} n$ consecutive hops is at most $1/n^4$. The result immediately follows.
	\end{proof}
	
	\paragraph{Carrying over Previous Results} 
	In the following we show that the result of \cref{lem:tree-bound} also holds when nodes follow the \emph{Intervals} protocol. That is, for some fixed node $w \in V_G$ we construct a tree as follows.
Assume  $w \in R_i$ and let $L_0 = \{w\}$ denote the root of this tree. The $j$-th level of the tree associated with $v$ is defined by $L_j = \{v ~|~ (v \in R_{i-j}\cap V_B )\wedge (\pi_v(k)\in L_{i-1}) \wedge (\forall \ell < k: (v,\pi_v(\ell))  \in \Fin) \wedge ((v, \pi_v(k)) \not \in \Fin )\}$. That is, $L_j$ is the set of nodes whose packets reach $v$ in exactly $j$ hops (for easier readability we neglect the fact that wrap-around might occur). 
Formally we show the following.

	\begin{lemma}
	\label{lem:interval-level-shrinking}
		Let $w \in V_G$ be arbitrary and assume that $w \in R_i$. Furthermore let $X_j = |L_j|$, where $L_j$ is the $j$-th level of the tree associated with $v$. Then, it holds that
	%	\[
			$E[X_{j+1}] \leq X_j \cdot \alpha$.
	%	\]
	\end{lemma}
	
	\begin{proof}
		Let $w \in V_B \cap R_i$.
		Assume that due to the failed edges by the adversary, $X_j$ many nodes of the interval $R_{i-j}$ route packets over $w$ to the destination.
		Consider now some node $v \in R_{i-(j+1)} \cap V_B$. According to the assumption of \cref{thm:interval} at most $\varepsilon \cdot I$ such nodes exist. Hence $v$ hits one of the $X_j$ nodes w.p.~smaller than $X_j / (I - |\Fin_{i-(j+1)}|)$. Now, since we have $|\Fout_{i-(j+1)}| = | R_{i-(j+1)} \cap V_B|$, we obtain that
		\[ 
			E[X_{j+1}] \leq |F_{i-(j+1)}^{(out)}|  \cdot  \frac{X_j}{I - |\Fin_{i-(j+1)}|} \leq \alpha \cdot X_j	\qedhere\] 
	\end{proof}
	
	With this we established a statement similar to the first part of \cref{lem:level-shrinking}. It is easy to see that the size of each level $X_j$ can be modelled with a sum of independent Poisson trials,  when constructing the tree level-by-level. Furthermore, \cref{lem:interval-paths} establishes the property of \cref{lem:tree-height} and since \cref{lem:interval-level-shrinking} guarantees that the levels shrink exponentially fast in expectation, the statement $X_i < C \log n$, $C$ large enough, follows by applying Chernoff bounds.
 That said, we established all necessary requirements and a simple repetition of the corresponding analysis allows us to reuse \cref{lem:tree-interval} and \cref{lem:tree-bound}.
	Summarizing, we get the following and conclude the proof of \cref{thm:interval}.
	\begin{corollary}
		Let $w \in V_G$ be a good node and let $\{X_i\}$  be defined as in \cref{lem:interval-level-shrinking}. Then it holds that $\sum_i X_i = O( \log n \cdot \log \log n)$ w.p.~$1-\polylog n / n^3$.
	\end{corollary}
	

%% file: hop.tex
\section{Further Reducing the Congestion}
\label{sec:hop}

In this section we present a third protocol, called \emph{Shared-Permutations}, that improves the bound of the maximum load observed in \cref{thm:destination} and \cref{thm:interval} under the assumption that the nodes share a common but randomized permutation over $V$. This could for example be achieved by computing parts of the routing tables starting from the same seed for the random generator, which is unknown to the adversary. While this assumption is indeed a weakness inherent to this protocol, it can be offset in case the adversary does not compromise one of the nodes directly: if all nodes manage to agree on a new permutation from time to time, this may invalidate previously obtained information by the adversary about the traffic flow.  We also assume that the packet headers are equipped with a hop field of size $O(\log \log n)$ bits, which is initially set to $0$ and may be accessed by the nodes of the network.

The \emph{Shared-Permutations} protocol works as follows. Again we consider an arbitrary but fixed destination $d$. Every node $v \in V$ is equipped with permutations $\pi_{0,d}, \pi_{1,d}, ... , \pi_{C_1,d}$ of all nodes $V \setminus \{d\}$, where $C_1$ is a value $O(\log n)$ to be specified later. Now, contrary to the \emph{3-Permutation} protocol, these permutations are assumed to be \emph{globally} agreed upon without being known to the adversary. Furthermore, each permutation is chosen u.a.r out of the set of all possible permutations. Additionally we assume that $v$ stores $C_2$ additional permutations $\pi_{v,j,d}$ on $V \setminus \{d\}$,  only known to $v$ itself and chosen u.a.r. Here $j  \in \{E_2,E_2 + 1, \dots , E_2 + C_2 +1\}$ for $E_2 = C_1 +1$ and $C_2$ is another value in $O(\log n)$.

Assume now that a packet $p$  with destination $d$ arrives at node $v \in V$ and denote its current hop counter by $h(p)$.
First of all, if the link $(v,d)$ is not failed the packet is directly forwarded to the destination.
Otherwise if $h(p) < E_2$, the node $v$ forwards it via a link $(v, v')$ where $v'$ denotes the node following $v$ in the global permutation $\pi_{h(p),d}$. In case this link is failed, $v$ raises the hop counter of $p$ to $E_2$ instead and forwards it to the \emph{first non-failed} edge according to $\pi_{v, E_2,d}$.
The case we did not consider yet is $h(p) \geq E_2$. In this case $p$ is routed over the first reachable partner in $\pi_{v,h(p),d}$. Finally, in every case, $h(p)$ is increased by one. A pseudo-code describing this algorithm is given in \cref{alg:hop}. The common global permutations  allow the flow to be distributed more evenly among the network, reducing the congestion to $O(\sqrt{\log n})$, even if $\Omega(n)$ edges are failed by the adversary.

\begin{figure}
\begin{algorithmic}[1]
\renewcommand{\algorithmicrequire}{\textbf{Input:}}
	\Require A packet with destination $d$ and hop count $h(p)$
	\If {$(v,d)$ is not failed} forward $p$ to $d$, set $h(p) \gets h(p) + 1$ and \algorithmicreturn \EndIf
	\If {$h(p) < E_2$} $v' \gets$ successor of $v$ in $\pi_{h(p),d}$
		\If {$(v,v')$ is not failed} forward $p$ to $v'$
		\Else \, $h(p) \gets E_2$.
		\EndIf
	\EndIf
	\If {$h(p) \geq E_2$} forward $p$ to first directly reachable node according to $\pi_{v,h(p),d}$
	\EndIf	
	\State $h(p) \gets h(p) + 1$
\end{algorithmic}
\caption{\emph{Shared-Permutations} protocol. Point-of-view of some node $v$}
\label{alg:hop}
\end{figure}

\begin{theorem}
\label{thm:hop}
Assume that the adversary is allowed to fail $\alpha \cdot n$ edges total, where $\alpha < 1$ is a constant\footnote{Just as in our first algorithm, $\alpha$ can be any constant  that lies in the range $0 < \alpha < (n-1) / n \approx 1$.}. When performing all-to-one routing to any destination $d$, the \emph{Shared-Permutations} protocol guarantees a maximum flow of
\[ 
	O( \sqrt{\log n})
\]
on any node (except $d$) and edge w.h.p.~
Additionally, no packet traverses more than $O(\log n)$ hops w.h.p.~
\end{theorem}

Assuming it is possible for the nodes to agree on common permutations that are not known to the adversary, the maximum load can be decreased by more than a factor $\sqrt{\log n}$  compared to the protocols in  \cref{sec:destination,sec:interval}. Note that this result breaks the $\Omega( \log n / \log \log n)$ lower bound of the \emph{3-Permutations} and \emph{Intervals} protocols.

Regarding space complexity, our nodes are required to store $O(\log n)$ permutations of $n$ nodes per destination. 
Therefore in the most simple case we require $O(n^2 \cdot \log^2 n)$ bits at most.
However, the same improvements as described in  \cref{sec:destination,sec:interval} can be made to store the permutations more efficiently and achieve a  memory complexity of $O(\log^3 n / \log(1/\alpha))$ bits per node.
Note that the protocol requires knowledge of the values $C_1$ and $C_2$, which can both be set to $5  \log_{1/\alpha} n$. 
These values are given in \cref{lem:hop-no-inner-max-hops} and \cref{lem:hop-S-maxhops}, together bounding the maximum number of hops any packet performs until it reaches the destination $d$ w.h.p.~
If $\alpha$ is not known to the nodes, then a slow growing function in  $\omega(\log n)$ can be used for $C_1$ and $C_2$, which comes at the cost of slightly increased memory complexity. 

\subsection{Analysis}
\label{sec:hop-analysis}
Throughout the analysis we consider a fixed destination $d$ and omit the corresponding index from all permutations. We use the notation defined in \cref{sec:destination-notation} and start by neglecting any failed edges in the set $\Fin$.
Next, we assume that each node sends $1$ packet with destination $d$ from each node $v \in V$. We consider the number of packets that have $i$ hops while still not having reached the destination $d$ and see that this set decreases exponentially fast. Additionally no packet traverses more than $C_1 < E_2$ many hops w.p.~at least $1-n^{-2}$. Therefore, without any inner edge failures, the only relevant permutations for our failover strategy are $\pi_{0}, ... ,\pi_{E_2 - 1}$.

Finally we account for the failures in $\Fin$ and make use of the permutations $\pi_{v,j}$. We consider the maximum load caused by the flows after reaching hop value $E_2$ separately and deduce that this value is $O(\sqrt{\log n})$ w.h.p.

\paragraph{Staying in Line}
We start by neglecting the failures in $\Fin$, i.e we assume first that $|\Fin| = 0$ and $|\Fout| < \varepsilon \cdot n$.
Furthermore, assume that every node $v \in V \setminus \{d\}$ starts sending a single packet to  destination $d$. Then, the number of packets that pass through some node $v$ is equivalent to the number of flows passing through $v$. Notice, that due to the global permutations, no node is visited by more than $1$ packet with the same hop value.
Consider the set of packets hop-for-hop and denote $H_i$ as the set of packets that reached hop $i$ at some point without reaching the destination.
Clearly $|H_0| = (n-1)$, $|H_1|=|V_B|$ and we can show the following.

\begin{lemma}
\label{lem:hop-shrinking-h}
	Assume $|\Fin| = 0$. Let $H_i$ denote the set of packets have not reached $d$ after $i$ hops. Then, for $|H_i| = \Omega(\log n)$ it holds w.p.~at least $1-n^{-4}$ that
%	\[ 
		$|H_{i+1}| \leq |H_i| \cdot \sqrt{\varepsilon}$.
%	\]
\end{lemma}

\begin{proof}
	Fix some $i < E_2$ and consider the set of packets $H_i$. Let $S_i$ be the set of nodes hosting the packets $H_i$. Clearly the nodes $v \in S_i$ hosting a packet (remember each node either hosts one or $0$ packets) are distributed uniformly across the network. This results from the fact that  $|\Fin| = 0$ and the permutations are chosen u.a.r.~from each other.
	
	The question is now what is the size of $S_i \cap V_B$, i.e. what is the number of packets in $H_i$ that can \emph{not} directly leave the network over a direct link to $d$ with hop $i+1$. As the nodes in $S_i$ are uniformly distributed over the network, consider the following process to determine the number of bad nodes in~$S_i$.
	
	 Enumerate the packets in $H_i$ as $\{p_1, ... p_{|H_i|}\}$ and assign to each $p_j$ some node $v_j \in V \setminus \{d\}$, chosen u.a.r.~\emph{without replacement}. We are only interested in counting the number of nodes $v_j$ that then lie in $V_B$. This can be modeled as an urn process, where we draw $|H_i|$ out of $n-1$ total balls, with $|V_B|$ of these balls being black. Answering the question of how many drawn balls are black yields $|S_i \cap V_B| = |H_{i+1}|$. We define the r.v.~$Y_j$ to model the $j$-th draw, where $Y_j = 1$ iff a black ball was drawn, and $Y_j = 0$ otherwise. Hence, we are interested in $|H_{i+1}| = Y := \sum_{j=1}^{|H_i|} Y_j$. Clearly the values $Y_j$ are not independent. However, $Y$ follows a hypergeometric distribution, which according to \cite{JP83}  is subject to the \emph{negative association} property. As stated in Theorem 3.1 of \cite{DP09} we may apply Chernoff bounds and since  $E[Y]  < |H_i| \cdot \varepsilon (1 + o(1))$, the result follows as long as $|H_i| > c \log n$ for a large enough constant $c$.
For the next hop a completely independent permutation is used. Therefore, the set $S_{i+1}$ is again uniformly distributed allowing this approach to be repeated.\end{proof}

From the exponential shrinking in \cref{lem:hop-shrinking-h}, it follows that no packet takes more than $O(\log n)$ hops to reach $d$. In \cref{appendix2} we show the following statement. 

\begin{lemma}
\label{lem:hop-no-inner-max-hops}
	Fix some packet $p$ with destination $d$. Then, assuming $\Fin=0$, it requires  at most $C_1$ steps to reach the destination w.p.~at least $1 - n^{-3}$. Here $C_1$ is a value bounded above by $5 \log_{1/\varepsilon} n$.
\end{lemma}

Recall the following invariant: when fixing some node $v$ and hop value $i$, the node $v$ receives \emph{at most} $1$ packet with such hop value. This leads to following result.

\begin{lemma}
\label{lem:hop-no-inner}
	Consider some node $v \in V$ and assume $|\Fin| = 0$. Then, if every node sends $1$ packet with destination $d$, $v$ is visited by packets
	%\[ 
		$O(\sqrt{\log n})$
	%\]
	times w.p.~at least $1 - n^{-3}$.
\end{lemma}
\begin{proof}
	Let $i^*$ be the first time such that $H_{i^*}$ reaches size $O(\log n)$. We start by showing that throughout hops $1 \leq i < i^*$ the node $v$ is visited by at most $O(\sqrt{\log n})$ packets in total. For this range of $i$, we know according to \cref{lem:hop-shrinking-h} that the size of $H_{i}$ decreases exponentially fast, i.e $|H_{i+1}| < |H_i| \cdot \beta$ for some constant $0 < \beta < 1$ that depends on $\varepsilon$. Fix now some node $v$ and let $Y_i =1$ iff $v$ receives a packet with hop value $i$ at some point, and $0$ otherwise. Similar as in the proof of \cref{lem:hop-shrinking-h}, we argue that the packets with hop value $i$ are distributed uniformly -- and independently of any earlier hops --  among the nodes of the network. Hence $P[Y_i =1] = H_{i} / (n-1)$ and we are interested in $Y = \sum_{i=1}^{i^*} Y_i$.
	Note that $P[Y_i =1] = H_{i} / n < \varepsilon \cdot \beta^i$ , where we wrote $n$ instead of $n-1$ for ease of readability. Then
	\begin{align*}
		P[Y = k] = \sum_{ \substack{S \subseteq \{1, ... O(\log n)\} \\ \text{with } |S| = k} } \big( \prod_{j \in S} P[Y_j = 1]  \cdot \\
		\prod_{\ell \in \{1 ... O(\log n)\} \setminus S} 1 - P[Y_\ell = 1] \big).
	\end{align*}
 The second product can be crudely bounded above by $1$. The first product reaches its maximum value for $S=\{1,...,k\}$ and is bounded by $ ( \varepsilon\beta \cdot \varepsilon\beta^2  \cdot ...  \cdot \varepsilon \beta^k)$. Therefore we get
\begin{align*}
		P[Y = k] 
		%&< {{C \cdot \log n}\choose{k}} \cdot  \varepsilon^{k} \cdot \beta^{k(k-1)/2} \\
		&< \left( \frac{e \cdot C \log n}{k} \right)^{k} \varepsilon^{k} \cdot \beta^{k(k-1)/2} .
	\end{align*}
	Now assume $k=C' \cdot \sqrt{\log n}$ for some sufficiently large constant $C'$. In this case
	\[
		P[Y = k] < O \left( \sqrt{\log n} \right) ^{C' \sqrt{\log n}}  \cdot \left( \frac{1}{n} \right)^5 .
	\]
	Clearly the first term lies in $o(n)$. When increasing $k$ further, the probability  only gets smaller. Applying the union bound over the  remaining $O(\log n) - C' \sqrt{\log n}$ larger values of $k$, we get  $P[Y \geq C' \sqrt{\log n}] < n^{-3}$.
Adding $1$ for the packet that was initialized on $v$ yields the result. 
\end{proof}

	Clearly this implies that both, the maximum node load and the edge load are $O(\sqrt{\log n})$ in the case of $|\Fin| = 0$.

\paragraph{Accounting for Inner Edge Failures}
We consider two copies, $S^{(out)}$ and $S$ of our initial graph in which the adversary failed at most $\alpha n$ edges according to its strategy. In $S^{(out)}$ we repair all failures $\Fin$, which results in ignoring inner edge failures just as described above. Again we consider the equivalent point-of-view of each node sending a single packet to $d$ instead of a consecutive flow. The idea in the following is to consider only $S^{(out)}$ and each time an inner-edge $(u,v)$ is chosen for communication that is failed in the original graph, the packet is copied and  placed with hop count $E_2$ on $u$ in $S$. This way $S$  contains  packets with hop count of at least $E_2$. The packets in $S^{(out)}$  however continue as if the edge was intact. 
Note that, by \cref{lem:hop-no-inner-max-hops}, $S$ w.h.p.\ only consist of packets that are redirected because of inner edge failures.
The idea behind the analysis is the following: Let $S^{(out)}$  run until all packets reached the destination $d$ and determine the number of packets starting in $S$. We then let the system $S$ run and it is easy to see that we can majorize the load some node $v$ receives in the original process by adding up the loads of $v$ in $S^{(out)}$ and $S$ respectively. This is because in $S^{(out)}$ we do not remove packets but copy them to $S$ instead.

In \cref{lem:hop-no-inner} we already established the load some node $v$ receives in $S^{(out)}$.
We start by determining the number of packets that are initialized in the system $S$.

\begin{lemma}
\label{lem:hop-S-packets}
Consider the number of packets $p$ that reach a load of $E_2$ at some point. Then, at most $O(\log n \cdot \sqrt{n})$ of them exist w.p.~at least $1 - 2 \cdot n^{-3}$.
\end{lemma}
\begin{proof}
First, consider some randomly selected permutation $\pi$ of the nodes $V \setminus \{d\}$. Define $X_i$ to be the r.v.~indicating whether the edge $(\pi(i), \pi(i+1)) \in \Fin$ and let $f_i$ denote the number of failed inner edges at $\pi(i)$, all for $1 \leq i \leq n-2$. 
One can see the construction of $\pi$ as follows. First, $\pi(1)$ is chosen at random, then we sequentially sample nodes without replacement to continue to the permutation and alongside determine the value of the $X_i$. When following this approach  $P[X_i = 1 | X_0, ... ,X_{i-1}] \leq \min \{ f_i / r(j) ,1 \} =: p_i$, where $r(j) := (n-1) - i$ denotes the number of remaining nodes that may be sampled by $\pi(i)$. Note that $p_i$ bounds $X_i$ independently of  any $X_j$ with $j < i$ as we crudely assume that all failed edges reach into nodes that are still open for sampling, i.e. nodes that are not already in the set $\{\pi(j) | 1 \leq j < i\}$. Let now $I_k := ( n/ \log^k n , n/\log^{k+1} n]$ and $L := \log n / \log \log n$. Observe, for $0 \leq k \leq L -1$ it holds that $\sum_{r(i) \in I_k} p_i \leq \min \{\log^{k+1} n , n/ \log^k n\}$. Then, 
 \[ 
	\sum_{i=1}^{n-2} p_i = \sum_{k=0}^{L-1} \sum_{r(i) \in I_k} p_i \leq \sum_{k=0}^{(1/2) L - 1} \log^{k+1} n + \sum_{k=(1/2)L}^{L-1} \frac{n}{\log^k n} .
 \]
 Both these sums can be represented by a geometric series of the form $ \sqrt{n} \cdot \sum_{k=0} \log^{-k} n$ when shifting the indices accordingly. It is easy too see that  their value can be bounded by $2 \sqrt{n} (1 + o(1))$ in total. If we now define the independent Bernoulli trials $X_i^*$ with $P[X_i^* = 1] = p_i$, then for $X^* := \sum_i X_i^*$ it holds that $E[X^*] < 2 \sqrt{n} (1 + o(1))$. According to Lemma 1.19 of \cite{AD11} $X = \sum_i X_i$ is majorized by $X^*$, which can be bounded using Chernoff bounds. The result follows by applying the union bound and considering that according to \cref{lem:hop-no-inner-max-hops} only inner edge failures can lead to a hop count of $E_2$ w.p.~$(1-n^{-3})$.
\end{proof}

As all packets in $S$ have hop count at least $E_2$, only the \emph{local} permutations $\pi_{v,j}$ are used as part of our failover strategy. While this leads to nodes possibly receiving multiple packets of the same hop value, the number of initial packets lies in $O(\log n \sqrt{n})$ only. We give a detailed proof to the following statement in \cref{appendix2}.

\begin{lemma}
\label{lem:hop-S-maxhops}
	\begin{enumerate}
		\item Fix some arbitrary packet in $S$. Then, it reaches the destination after at most $C_2$ hops for $C_2 < 3 \log_{1/\alpha} n$ w.p.\ at least $1-O(n^{-3})$.
		\item Each node in $S$ is reached by at most $O(\sqrt{\log n})$ packets in total and w.p.~at least $1 - O(n^{-3})$.
	\end{enumerate}	 
\end{lemma}

We established that in both systems, $S^{(in)}$ and $S$, each node has a load of $O(\sqrt{\log n})$. \cref{thm:hop} follows accordingly.

%% file: simulations.tex
\newcommand{\bibd}{$\mathtt{A\text{-}CASA}$\xspace}
\newcommand{\squareone}{$\mathtt{Square1}$\xspace}

\newcommand{\deter}{$\mathtt{A\text{-}Det}$\xspace}
\newcommand{\rand}{$\mathtt{A\text{-}PRNB}$\xspace}

\newcommand{\intervald}{$\mathtt{Interval\text{-}D}$\xspace}
\newcommand{\permd}{$\mathtt{ThreeP\text{-}D}$\xspace}

\newcommand{\intervalsid}{$\mathtt{Interval\text{-}ID}$\xspace}
\newcommand{\permsid}{$\mathtt{ThreeP\text{-}ID}$\xspace}

\section{Simulations}
\label{sec:simul}

To complement our theoretical analysis, we compare an adapted version of our protocols against other state-of-the-art local  failover strategies \cite{DBLP:journals/ton/ChiesaNMGMSS17,infocom19casa} in the widely deployed Clos datacenter topology \cite{clos,singh2015jupiter}. More precisely, we consider the special case of the Clos topology with $3$ layers, sometimes simply referred to as fat-tree. All source code used to derive the results in this section can be found on GitHub \cite{simsource}.

For our experiments we considered 8 different protocols in total, which can be described as follows.

\noindent{\textbf{Our own protocols:}} Abbreviated by \permd, and \intervald we consider variants of our \emph{3-Permutations} and \emph{Intervals} protocol from \cref{sec:destination,sec:interval}, adapted to the Clos topology. Additionally, we consider two further variants of these protocols denoted by \permsid and \intervalsid. The ending -$\mathtt{ID}$ indicates that for these protocols, we select the permutations which are used for forwarding not only depending on the destination of the packet but also the inport from which the packet arrives. We employ these additional variants of our protocol as the reference protocols require support for destination and inport-based forwarding. A detailed description of these protocols is given in  \cref{sec:clos-adaptation}.

\noindent{\textbf{Related Approaches:}} We consider the state-of-the art local failover protocols \emph{DetCirc}, \emph{PRNB},  \emph{CASA} and \emph{SquareOne} \cite{infocom19casa}. Throughout our experiments, we refer to them as \deter, \rand, \bibd and \squareone, respectively. The first two protocols are slightly modified versions of the \emph{HDR-Log-K-Bits} and \emph{Bounced-Rand-Algo} originally presented in  \cite{DBLP:journals/ton/ChiesaNMGMSS17}, and versions of all of these protocols have also been evaluated using simulations in~\cite{infocom19casa}.
The first three protocols have in common that they are so-called \emph{arborescence-based} routing protocols. To route a packet toward some destination $d$ in a topology which is $\ell$-connected, the protocols use a set of pre-computed sub-trees $\{T_0, T_1, ... , T_{\ell-1}\}$ called arborescences. Each such tree has $d$ as its root, consists of all nodes in the topology and does not share any directed edges with other trees in the set. Packets then start on some arborescence $T_i$ and are routed along the edges in the tree until they hit the destination. In case the packet arrives at a failed edge, another arborescence $T_j$ with $j\neq i$ is selected along which the packet may continue its path towards $d$. This procedure is repeated until the packet arrives at $d$. The aforementioned protocols differ in the way this alternate arborescence $T_j$ is selected in case a failed edge is encountered. In \deter, $j$ is selected deterministically and set to $(i+1) \mod \ell$. In \rand, $j$ is selected uniformly at random out of $\{0,1,..., \ell-1\} \setminus \{i\}$. Finally,  \bibd uses a sophisticated pre-computed matrix which is constructed via so-called balanced incomplete block designs (BIBD) to deterministically select $j$. 

The \squareone protocol operates differently. In this protocol, packets with destination $d$ and source $s$ are routed over one of the $\ell$ shortest edge-disjoint paths from $s$ to $d$. At first, the packet attempts to follow the shortest such path to reach $d$.
However, in case a failed edge is encountered, the packet needs to traverse back to $s$ and follow the next-shortest path instead. This is repeated until the destination $d$ is reached.

\subsection{Engineering Protocols for the Clos Topology}
\label{sec:clos-adaptation}

In this section we discuss the required adaptation of our protocols to be employed in the Clos topology as well as the computation of the arborescences and edge-disjoint shortest-paths required by the related protocols. We start with a short definition of the Clos topology, which is required to explain the required modifications to our own protocols.

The Clos topology with 3 layers consists of $k/2  \cdot k/2 + k\cdot k = \Theta(k^2)$ nodes (or routers), each having at least $k$ ports. These nodes are partitioned into $k/2$ \emph{blocks} and $k$ \emph{pods}, which we assume to be numbered from $1$ to $k/2$ and $1$ to $k$, respectively. Each block contains exactly $k/2$ many nodes, which we again assume to be numbered from $1$ to $k/2$. Each pod consists of two sets of $k/2$ nodes each. The first set we call the $\emph{top}$ nodes while the second set we call $\emph{bottom}$ nodes. Bidirectional links are inserted such that in each pod, the top and bottom nodes form a complete bipartite graph. Additionally, the $i$-th top node in each pod is connected to all nodes in the $i$-th block and vice versa.
Endpoints using this communication infrastructure are connected at the remaining $k/2$ open links at each bottom node. Throughout our experiments, we focus on forwarding flows which have bottom nodes as source and destination.

\noindent{\textbf{Adapting our Own Protocols:}} We start with an explanation of the \intervald protocol.  Note that this adapted protocol has been analyzed theoretically after we performed our first experiments in a follow-up paper \cite{BES21}. First we need to partition the Clos topology into more fine-grained pieces. The top as well as the bottom nodes are split into $K:= \lfloor\log(k) \rfloor$ consecutive partitions, deviating in size by at most $\pm 1$. Similarly, the $k/2$ block nodes in each block are also split into $K$ many intervals. Furthermore, consider nodes in the $b$-th block. Each such node is connected to the $b$-th node in the top layer of every pod. We also assume that, for each block, the set of such top nodes is partitioned into $K$ what we call $\emph{vertical intervals}$.
In the remaining description, when we say that some node $v$ forwards a packet with destination $d$ to a random node in an interval, then we assume that the random selection follows the approach described in \cref{sec:interval}. That is, the node $v$ consults a random permutation of all nodes in this interval (one such permutation is precomputed for each destination $d$) and then forwards the packet to the first node $u$ in this permutation, such that the link $(v,u)$ is not failed.
When following the \intervald protocol, the forwarding rules for a packet with destination $d$ arriving at a node $v$ then depend on whether $v$ is a block, top or bottom node. First, assume that $v$ is a bottom node in the $p$-th pod. Then, if $v=d$ nothing needs to be done. Otherwise, let $j$ denote the interval of $v$. In such case, $v$ forwards the packet to a random top node of pod $p$ that also lies in the $j$-th interval. Second, if $v$ is a top node in some pod $p$ (in more detail, let $v$ be the $i$-th top node in $p$). Then, if the destination also lies in $p$, the node $v$ attempts to forward the packet directly to $d$. In case this link is failed, it instead forwards the packet to a random bottom node of $p$ in interval $(j+1) \mod K$. This way, as soon as the packet lies on some node in the pod of the destination, it will ping-pong between bottom and top nodes until it reaches a top node whose link to $d$ is not failed. However, if $v$ is not in the pod of the destination, then it forwards the packet to a random node in the $j$-th interval of the $i$-th block. In the latter case $v$ is a node in some interval $j$ of a block (let it be the $b$-th block). Each node in block $b$ is connected to the $b$-th top node in the pod of the destination. In case the link to this top node is not failed, $v$ forwards the packet over this link. Otherwise, $v$ forwards the packet to a random top node in the vertical interval $(j+1) \mod K$.

Following above description, the \permd protocol, which can be seen as an adaptation of the \emph{3-Permutation} protocol, is now easy to define. This definition is very similar to the \intervald protocol except for two differences: First, we set $K=1$, i.e., we don't split the block, top, or bottom nodes into further intervals. Second, we assume that each node $v$ does not only store one permutation but $6$ permutations per destination $d$. Depending on the hop count of the arriving packet one of these permutations is selected. More precisely, the $i$-th permutation with $i \in \{1,2,...,5\}$ is consulted for packets with hop count $[(i-1) \cdot \lfloor \log k \rfloor~,~i\cdot \lfloor \log k \rfloor]$. The $6$-th permutation is used for packets with any larger hop count. The reason that we use $6$ permutations instead of the $3$ as defined in \cref{sec:destination} is that empirical results indicated that it is beneficial to switch permutations frequently. This, however, requires the employment of additional permutations to avoid the creation of permanent forwarding loops. 

Finally, the two additional variants called \permsid and \intervalsid. They follow the exact same definition as their protocols of similar name with only one exception. We now assume that nodes store additional sets of these randomly generated permutations: one per combination of possible destination address \emph{and} inport. In contrast to the basic \permd and \intervald protocols, which select the permutations used for forwarding solely depending on the destination address, we assume this selection to be performed randomly for each pair of destination and inport.

\noindent{\textbf{Employing Related Approaches:}}
While the related protocols we consider are applicable to general graphs, they require the pre-computation of sets of arborescences as well as edge-disjoint shortest paths for every possible destination $d$. In order to compute the set of arborescences required by the first group of protocols, we employed the \emph{round-robin} approach with swaps presented in \cite{dsn19}. Even with this efficient approach, the calculation of these arborescences for a single destination node $d$ of the Clos topology with $k=80$ required more then 20 hours on the MACH-2 supercomputer (\url{https://www3.risc.jku.at/projects/mach2/}). Some preliminary tests showed that the computation time is roughly proportionate to $k^6$, which prevented us from using larger topologies in our tests. Similarly, the edge-disjoint shortest paths required for the \squareone protocol took more than $15$ minutes to compute for a single destination $d$. To avoid costly recomputation of these structures for multiple destination nodes, we only computed them once for some fixed destination $d$ on the bottom layer. We then applied an isomorphism to map these structures towards the remaining possible destination nodes on the bottom layer.

\subsection{Experiments}

We conducted two different types of experiments, both in the Clos topology with $k=80$ (consisting of 8000 nodes).

\noindent\textbf{Experiment 1: Performance under all-to-one model.} In the first experiment we examine the performance of the fast rerouting protocols under the all-to-one traffic pattern. Results are given in \cref{fig:all-to-one}. Each simulation was started by first failing a $p$ fraction of random edges. We then select a random node on the bottom layer and let each other bottom layer node send one unit of flow towards this destination. After the routing procedure is complete, we measure the maximum edge load as well as the average amount of hops required by any flows to reach the destination. To obtain the results for the plots in  \cref{fig:all-to-one} we perform these simulations for increasing values of $p$ (ranging from $p=0.0$ to $p=0.2$ in steps of $0.02$), repeat the simulation for each $p$ value $40$ times and report the average of the resulting maximum edge load and hop values. 

\begin{figure}[tb]
    \centering
    \includegraphics[scale=0.6]{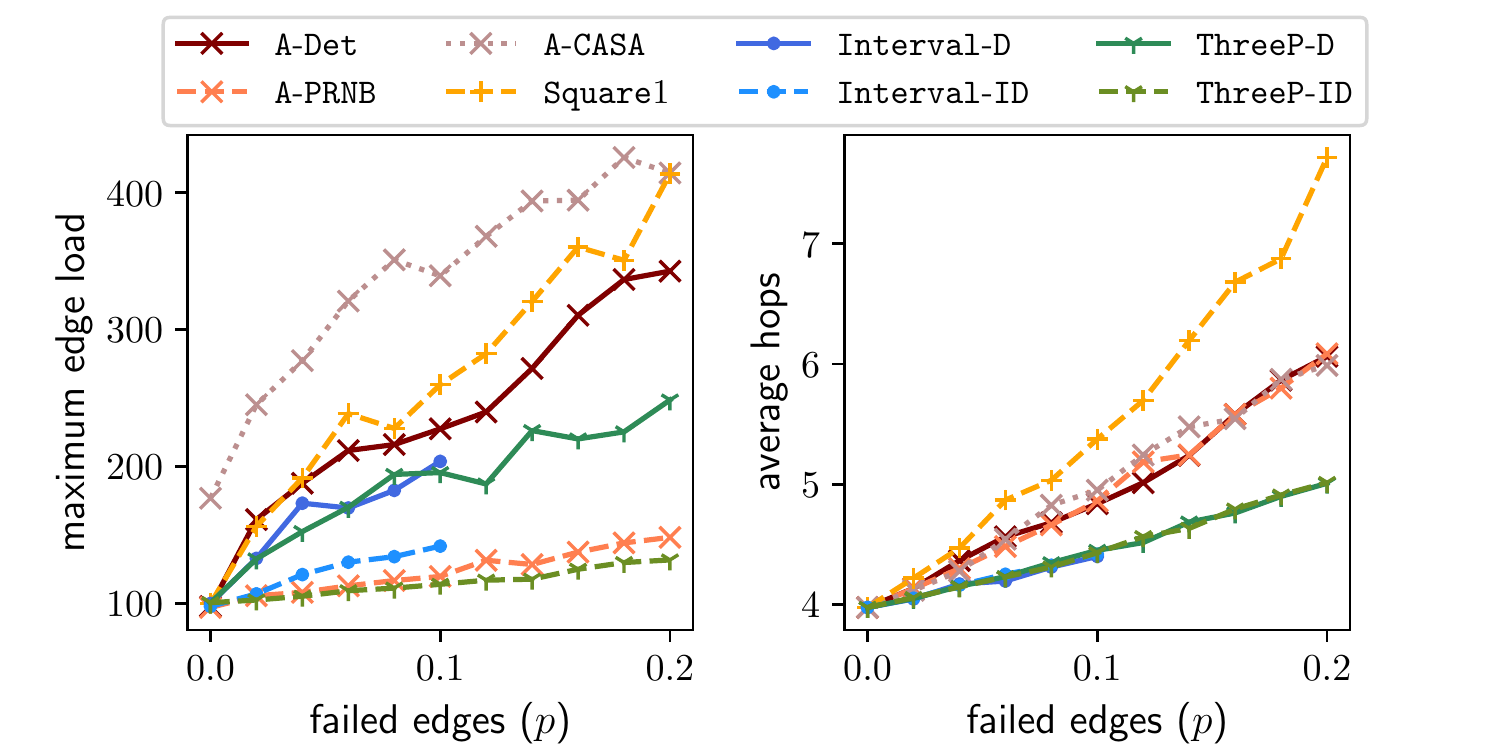}
    \caption{Average maximum edge load and average amount of hops required to reach the destination when performing \emph{all-to-one} routing.}
    \label{fig:sim-all-to-one}
\end{figure}

As we can see in the plot on the left-hand side of \cref{fig:sim-gravity}, all protocols besides \bibd accumulate similar loads in case no edges are failed. We suspect that the worse behavior of \bibd stems from the fact that the Clos network we consider is $40$-connected. As explained in \cite{infocom19casa} this protocol works best when this value is a prime power. In our setting, this led to multiple arborescences being (almost) completely unused, which causes the load to be distributed unevenly. It is also important to note that the results for our \intervald and \intervalsid protocols are reported till $p=0.1$. We do this because a higher amount of edge failures causes forwarding loops to be created or prevents the routing strategy from working (nodes get disconnected from all nodes in the adjacent interval). We emphasize that this is related to the small value of $k=80$ we consider in these experiments, which leads to intervals of size only $6$ (see the description of the adaptation of our protocol in \cref{sec:clos-adaptation}).   For values of $k>200$ we could not observe this behavior even when failing a $p>0.25$ fraction of all edges. When increasing the amount of failed edges in the system, only the randomized \rand protocol is able to compete with our \permsid protocol, which further illustrates the strength of randomized approaches when dealing with edge failures.

When looking at the right-hand side of \cref{fig:sim-all-to-one} we can see the average number of hops required for packets to reach the destination. There we can observe three regimes. First, we have the \squareone protocol which performs the worst. While it starts from a near-optimal average hop count of roughly $4$ (note that almost all source nodes in our all-to-one routing approach are $4$ hops away from the destination), it increases more rapidly than the other approaches. This is an inherent weakness of this protocol, as each time a packet with source $s$ encounters an edge failure on the way to $d$, it goes all the way back to $s$ and attempts another route. In the second regime, we can observe all the aborescence based approaches. We think the reason that these protocols perform worse than our protocol is the following: assume that a packet traversing some arborescence $T_i$ encounters a failed edge while being at node $v$. It will now continue from node $v$ in another arborescence $T_j$. Now, it is possible that the position of $v$ in $T_i$ is much closer to the destination than in $T_j$ and hence, this packet possibly needs to traverse a long path inside $T_j$ even though it was close to the destination before switching to this tree. In contrast, our protocols avoid sending packets on a long detour. If a packet resides at distance $x$ from the destination and cannot proceed closer due to a failed edge, it is forwarded between nodes in distance $x$ and $x+1$ until it is able to move to a node in distance $x-1$.

\begin{figure}[tb]
    \centering
    \includegraphics[scale=0.6]{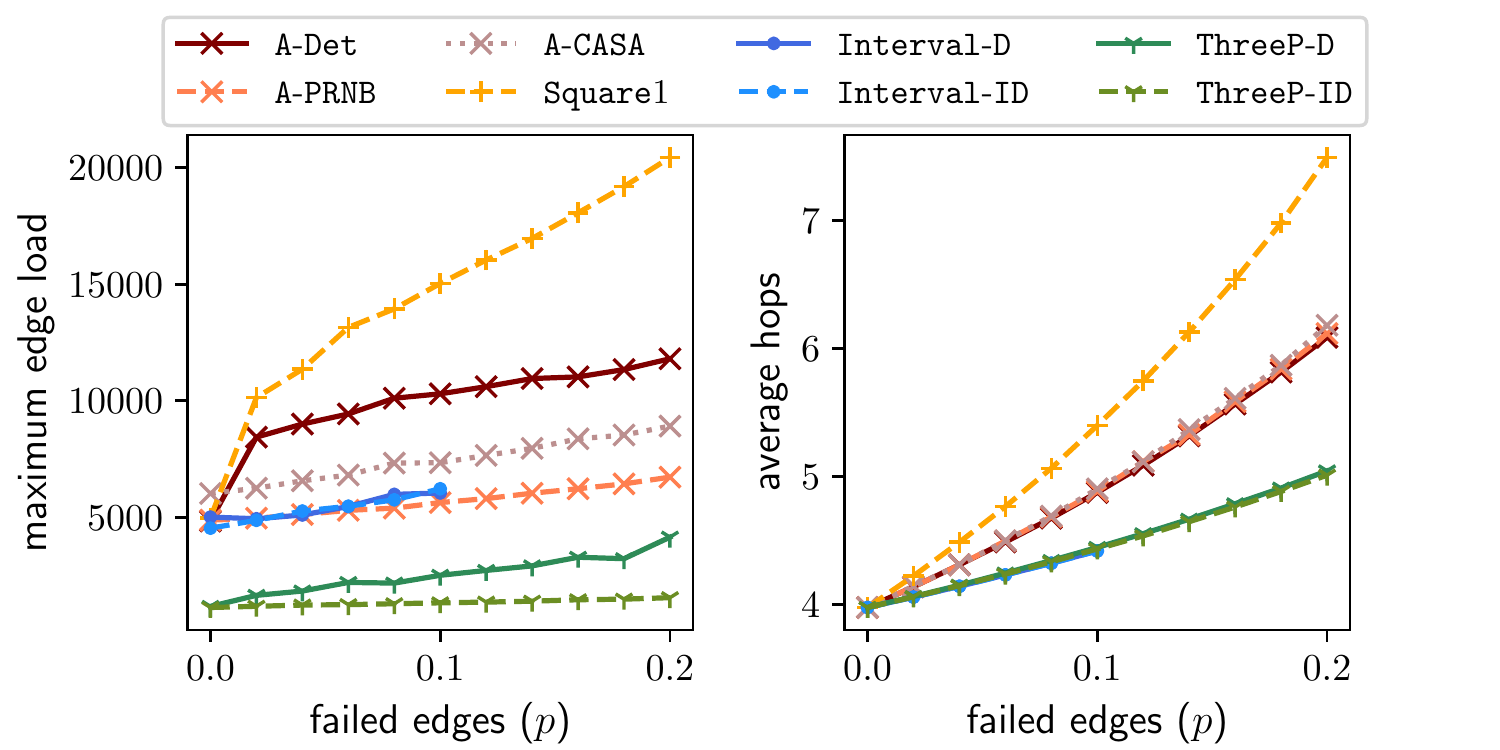}
    \caption{Average maximum edge load and average amount of hops required to reach the destination when performing \emph{gravity} routing.}
    \label{fig:sim-gravity}
\end{figure}

\noindent\textbf{Experiment 2: Performance under gravity model.} We next consider the performance of the protocols under a \emph{gravity model} \cite{gravitymodel}, describing the demands between any pair of two nodes on the bottom layer of the Clos topology. We set the parameters of this model such that, in expectation, the demand between each such pair of nodes is $1$. Throughout the routing process we then send a flow from each bottom node to every other bottom node (all-to-all). This flow is assigned a weight corresponding to the demand. For this experiment we also generalize our notion of node and edge load: load is now defined as the sum of all weights of flows that cross the node or edge.  Besides these changes, the steps for generating the results of our second experiment are the same as in the first.

As we can see on the left-hand side of \cref{fig:sim-gravity} our \permd and \permsid protocols achieve lower maximum load than all other approaches (including our \intervald and \intervalsid approaches). We suspect that this is because in the \permd and \permsid protocols, packets are forwarded according to random permutations that span over a larger amount of nodes. In the \intervald and \intervalsid the selection of possible forwarding partners in any hop is more constrained. This advantage is further emphasized on in the \permsid protocol as in this protocol nodes are only guaranteed to make the same forwarding decisions for flows that arrive from the same inport and have the same destination. This makes it less likely for multiple flows to follow the same path.
When it comes to the protocols of related work, we suspect the higher load to stem from the selection of the set of arborescences and also the selection of the set of shortest-paths for \squareone. There exists some discussion around the efficient construction of a good set of aborescences for a fixed destination \cite{dsn19}. However, it seems to be an open question how to select such sets for multiple destinations with the goal of optimizing load in many-to-many traffic patterns.
We found that these structures need to be sufficiently edge-disjoint from those used for other destinations. Otherwise, load imbalances occurred when performing gravity routing. In particular, we encountered a ``bad'' set of arborescences, which lead to an maximum node load of more than 200,000 in all arborescence-based protocols even if no edge is failed. In contrast, a better set of arborescences lead to a maximum node load of at most 40,000.
For all of our experiments, we employed the best sets of arborescences we could create in order to minimize the load values w.r.t. all the approaches we considered.

On the right-hand side of \cref{fig:sim-gravity} we can see the average number of hops required by the flows to reach their destination. From the point of view of any bottom node, the Clos topology looks exactly the same. Because of this inherent symmetry, the average hop values are very similar as in our all-to-one routing experiments. The reasoning for the three regimes of average hop values that can again be observed in this setting is the same as in the all-to-one experiments.

\noindent{\textbf{Takeaway:}}
We observe in our experiments that the \permsid protocols guarantee, both, the lowest maximum load as well as average hop count in both experiments. In case it is possible to match the source address, destination address, input and hop count in the packet header it is advisable to use this protocol. In case a slim set of forwarding decisions is required it makes sense to consider the \intervald protocol. It still outperforms all approaches of related work considering the average hop count and only gets outperformed by \rand when it comes to ensuring limited load. However, if resilience against a large amount of failures is required, then it is only advisable to utilize this strategy in topologies of sufficiently large size.

%% file: conclusion.tex
\section{Conclusion and Future Work}
\label{sec:conclusion}
\label{sec:failures}

In our work we considered three different local failover protocols.
Starting with the  \emph{3-Permutations} protocol, we presented a protocol which guarantees a load of at most $O(\log^2 n \log \log n)$ w.h.p.~even if $\alpha \cdot n$ edges are failed for some constant $0 < \alpha < 1$.
Next, we presented the \emph{Intervals} protocol. While this protocol comes with slightly lower theoretical resilience of $O(n / \log n)$, it can also be used in settings where the hop count in the packet header cannot be matched.
It achieves a maximum load of $O(\log n \log \log n)$ w.h.p. 
Finally, we presented a third approach, the \emph{Shared-Permutations} protocol, which is mostly of theoretical interest. In case the nodes have access to some shared permutation, we show that the maximum load can be reduced to at most $O(\sqrt{\log n})$ w.h.p.

We also adapted two of the above approaches to the Clos topology with 3 layers and performed emperical tests. These tests indicate that our protocols ensure a low edge load in this more practical setting as well.
The \permsid variant of the \emph{3-Permutations} protocol even outperforms all related approaches when it comes to the maximum edge load. Additionally, all our adapted protocols ensure that packets reach their desired destination in less hops than in related approaches in case multiple edge failures occur.
It remains an open question whether the above protocols can also be adopted to more general topologies.

Throughout all our theoretical results we assumed an \emph{oblivious} adversary which selects the set of failed edges. However, some of our algorithms can easily be extended
to deal also with \emph{more adaptive adversaries}: to defeat adversaries
who aim to infer network-internal loads (e.g., leveraging physical access or
using tomographic techniques), we can simply regenerate 
random permutations periodically.
That is, the \emph{3-Permutations} and \emph{Intervals} algorithms have the attractive property 
that they allow to regenerate
such permutations quickly, locally, and without coordination:
each node can independently regenerate random numbers
over time to enhance security. Note, this also allows our algorithms to recover if the low probability event 
occurs, in which higher loads than the ones specified in our theorems emerge.

There are also slight variations of our failure model which we did not fully cover in our analysis. For example, the case of a \emph{lower amount of edge failures}.
In our analysis, we assumed that the adversary destroys
up to either linear or $O(n / \log n)$ many edges. We believe that a lower amount of edge failures affects the performance of the algorithms as follows. 
If $n^{1-\delta}$ edges are destroyed for some constant $\delta >0$ it can be shown (by a slightly adapted repetition of the existing analysis) that all three of our algorithms guarantee w.h.p.\ a maximum load of $O(1)$ on most of the nodes. 

Finally, it may also make sense to analytically consider the case of \emph{randomly selected edge failures} instead of assuming the existence of a malicious adversary. In the context of hardware failures or power outages, it might make sense to model the set of failed edges as selected u.a.r. out of all $\sim n^2 / 2$ edges.  If only a small amount of edges is failed in such a way (i.e., $O(n)$), then all but $O(\log n)$ nodes may forward their flows directly to the destination w.h.p. It can be shown that all our algorithms then induce a congestion of $O(1)$ w.h.p. The question about a higher amount of edge failures and also the expected resilience against failures of such type remains an open question.

%% file: appendix.tex
\appendices

\section{}
\label[appendix]{appendix}

\begin{proof}[Proof of \cref{lem:level-shrinking}]
We show by induction on $i$ that with probability $p_j > (1- i \cdot  n^{-4})$ it holds for all $j\leq i$ that  $X_j < C \log n$. Clearly w.p.~$1$ it holds that $X_0 = 1  < C \log n$. Now, consider step  $(i+1)$ and observe that by \cref{lem:dest-level-binomial} 
\begin{align}
\label{eq:level-shrinking}
	E[X_{i+1}&|X_0, ... , X_{i}  < C \log n] \nonumber \\
	&< (|V_B| - \polylog n) \frac{X_i}{|V'| - \polylog n} \nonumber \\
	&< \varepsilon \cdot X_i \cdot \left(1 + \frac{\polylog n}{n} \right),
\end{align}
where we used the induction hypothesis and that $i< \log^2 n$ as well as $|V_B| \leq \varepsilon n$.
By \cref{lem:dest-level-binomial}, $X_{i+1}$ follows a binomial distribution. Therefore we apply Chernoff bounds with $\delta \approx \varepsilon^{-2} - 1$ (c.f \cref{sec:destination-notation}) and obtain that $P[X_{i+1} \geq C \log n~|~X_1, ... , X_i < C \log n] < n^{-4}$ for large enough constant $C$. Hence we established the desired property w.p.~$p_{i+1} \geq p_i \cdot (1 - n^{-4}) > (1 - (i+1)n^{-4})$ and conclude the induction. Using that $P[X_{i+1} \leq n] = 1$, for $i < \log^2 n$ we can bound $E[X_{i+1}] < E[X_{i+1}~|~X_0, ... ,X_{i} < C\log n] + n \cdot i \cdot n^{-4}$, since $P[X_i < n] = 1$. This, together with (\ref{eq:level-shrinking}), yields the first statement of the lemma.
\end{proof}
\medskip

\begin{proof}[Proof of \cref{lem:dest-with-inner-edges}]
According to \cref{lem:dest-log-inner-failures} we know that the failure of inner edges causes  $O(\log n)$ nodes to redirect their incoming packets via another path than given in $G'$. That is, for some node $v \in V_B$ the outgoing edge $(v , \pi_v(1))$  is replaced by some $(v, \pi_v(j))$ where $\pi_v(j)$ is the first non-failed link in $v$'s permutation. Note that this corresponds to relocating the whole subtree rooted in $v$ over to $\pi_v(j)$. Such a subtree will be called a \emph{relocated subtree} in the following.

There are now two major points to be checked. First,  two or more of these relocated subtrees may connect to each other, potentially causing a new cycle to be created. According to \cref{lem:tree-bound} and \cref{lem:tree-bound-cycle} the size of these subtrees may not exceed $O(\log n \cdot \log \log n)$. As the adversary can fail at most $\gamma n$ edges at $v$, a relocated subtree hits another relocated subtree w.p.~at most $O( \polylog  n / n)$. Since only $O(\log n)$ relocated trees exist, it is easy to see that each such subtree is hit by at most $O(1)$ other relocated subtrees, and at most a chain of length $3$ of subtrees may exist. Furthermore at most $O(1)$ cycles are formed this way, each consisting of a total of $O(1)$ relocated subtrees. All these statements hold w.h.p.~

The second question is whether the subtrees docking onto another component that already existed in $G'$ may increase their size substantially. Now, the adversary may only fail $\gamma n$ inner edges, hence the probability that a relocated subtree hits some arbitrary but fixed component is $O( \polylog n / n)$ as well. Again, using for example the PDF of the binomial distribution, it is easy to see that no component is hit by more than $O(1)$ such relocated subtrees directly w.h.p.~Since each of these subtrees are of size $O(\log n \cdot \log \log n)$, the size of the components does not change asymptotically.

Finally, it follows from \cref{lem:tree-height} and \cref{lem:tree-bound-cycle} that no relocated subtree can be of height larger than $C' \log n$. In the worst case a structure of size $C' \log n$ is docked on by a chain of $3$ relocated subtrees. Hence, the paths of packets not trapped in cycles is elongated to at most $4 C' \log n$.
\end{proof}
	\medskip
\begin{proof}[Proof of \cref{lem:dest-cycle-load}]
	We start with the proof of the first statement. Per assumption each component is entered by $O(\log n \cdot \log \log n)$ flows in total. These flows travel through the component until they either hit a cycle or some good node.  As every node in $G^{''(i)}$ has out-degree of either $0$ or $1$, each flow can hit a node $v$ that lies outside the cycle at most once.  Therefore, no such node receives more than $O(\log  n \log \log n)$ load in total.

	For the second statement, consider a flow that starts at a component which contains a cycle. As soon the flow reaches the cycle, it causes the accumulation of one unit of load at each node inside the cycle per turn. In the worst-case, the cycle is of length $O(1)$ and the flow accumulates $O(\log n)$ load at any node on the cycle before it exits $G^{''(i)}$. Together with the assumption of $O(\log n \log \log n)$ initial flows being present, the result follows.
\end{proof}
\medskip
\begin{proof}[Proof of \cref{lem:dest-cycle-load-assumption}]
	According to \cref{lem:dest-with-inner-edges},w.h.p., the only way for flows to enter $G^{''(i+1)}$ is if they spun in a cycle in $G^{''(i)}$. Denote now $S$ as the set of cycle nodes. Then, the flows enter $G^{''(i+1)}$ at the positions of these nodes $S$. According to \cref{lem:dest-with-inner-edges}, $|S| = O(\log^2 n)$, and we first consider where the set of nodes $S$ is located in $G^{'(i+1)}$. As the permutations are independent, the set $S$ is distributed uniformly among the nodes in $G^{'(i+1)}$. When sequentializing the placement of the set $S$, the probability to hit some fixed structure in $G^{'(i+1)}$ is $O(\polylog n / n)$ independently. As $|S| = O(\log^2 n)$ it is easy to see that at most $O(1)$ nodes of $S$ are located in the same structure w.h.p.~
	
Next, we account for inner edge failures and consider where these $S$ cycle nodes are located in $G^{''(i+1)}$. 
We already established that the inner failures cause up to $O(\log n)$ subtrees to be relocated among the structures in $G'$.
In the proof of \cref{lem:dest-with-inner-edges} we argued that $O(1)$ many subtrees relocate to the same structure and only $O(1)$ of them combine together to a new type of structure that contains a cycle. Clearly, each subtree contains only $O(1)$ cycle nodes as well. Therefore, also in $G^{''(i+1)}$ all components contain $O(1)$ nodes of $S$. Note that according to the assumption in \cref{lem:dest-cycle-load}, the structure of any $v \in S$ was entered by at most $O(\log n \log \log n)$ flows in $G^{''(i)}$. Clearly, only a flow that entered at some point can exit. Therefore each node $v \in S$ serves as entry point for at most $O(\log n \cdot \log \log n)$ flows in $G^{''(i+1)}$.	
\end{proof}

\medskip
\begin{proof}[Proof of \cref{lem:g3-enough}]
	To show this statement we follow the path a packet $p$ takes starting from some fixed $v \in V_B$. Similar to the proof of \cref{lem:dest-path-length} we look at the path starting at $v$ and uncover the edges which are traversed by packet $p$ in $G^{''(1)}$. Along the lines of the proof of \cref{lem:dest-path-length} we get that the resulting path forms a cycle with w.p.~at most $4 \log_{1/\alpha} (n) / n$. If this event indeed occurs, the packet $p$ will continue to traverse this cycle until it reaches a hop count of $C_1$. From this point on, the  $\pi_v^{(2)}$ is consulted to forward the packet $p$. By construction of the graph $G^{''(2)}$ it follows that we can now use this graph to  describes pack the packet $p$ takes in the next $C_1$ hops. A repetition of above argument again yields that $p$ is forwarded in a cycle in $G^{''(2)}$  w.p.~ at most $\polylog n / n$. The same argument can be applied a third time w.r.t $G^{''(3)}$ in case $p$ is also trapped in a cycle in $G^{''(2)}$. 
	Because $C_1$ is chosen such that the packet $p$ reaches destination $d$ unless it is stuck in such a forwarding loop (see \cref{lem:dest-with-inner-edges}), it follows that $p$ will reach the destination unless it is stuck in forwarding loops in $G^{''(1)}, G^{''(2)}$ and $G^{''(3)}$. As each of these graphs is induced by independent and randomly generated permutations, it follows that the probability for this event is at most $O(\polylog n / n^3)$.
\end{proof}

%%%%%%%% Start of Proofs of Hop-Protocol section

\section{}
\label[appendix]{appendix2}

\begin{proof}[Proof of \cref{lem:hop-no-inner-max-hops}]
	The exponential shrinking in  \cref{lem:hop-shrinking-h}  implies that at most $R := O(\log n)$ packets remain after $2\log_{1/\varepsilon} n = O(\log n)$ steps w.h.p.~
	Assume that at this point we are still using global permutations as part of our failover strategy. Similar as in the proof of \cref{lem:hop-shrinking-h} this set $S$ of nodes hosting these $R$ packets is uniformly distributed among the graph. Now, fixing some packet $p \in R$ we determine the the probability that it resides on a node $v \in V_B$. Again, there exist dependencies between the packets as each node $v \in V_B$ can host at most $1$ packet. At this point however only $O(\log n)$ packets $\{p_1 , ... , p_{O(\log n)}\}$ remain. Therefore any packet $p_i$ resides on a node $v \in V_B$ w.p.~at most
	\[
		\frac{|\Fout|}{n -1 - O(\log n)} < \varepsilon \left(1 + o(1) \right)
	\]  independently.
Hence, the packet $p \in R$ reaches the destination after $3 \log_{1/\varepsilon} n$ further steps w.p.~$1-n^{-3}$.
\end{proof}
\medskip
\begin{proof}[Proof of \cref{lem:hop-S-maxhops}]
	In \cref{lem:hop-S-packets} we established that at most $O(\log n \cdot \sqrt{n})$ packets start in the system $S$. All of these start with hop count $E_2$.  According to \cref{lem:hop-no-inner} at most $O(\sqrt{\log n})$ packets are initiated by the same node $v \in S^{(out)}$. In the following we call a set of packets being at the same node with the same hop counter a (packet-)\emph{bundle}. We can crudely assume that $O(\log n \cdot \sqrt{n})$ packet bundles of size less than $O(\sqrt{\log n})$ are distributed among the nodes of $S$ initially and thereby upper bound the accruing load. 
Each time, the next hop of a packet $p$ is at a node $v \in V_B$ w.p.~at most $| \Fout | / ((n-1) - |\Fin|) < \alpha (1 + o(1))$.  Hence, no packet performs more than $3 \log_{1 / \alpha} n$ hops as on each hop a new independent permutation is used.
	
	Next we analyze our process hop-by-hop, starting with hop count $E_2$.
When sequentializing the target selection of the bundles, each bundle hits at least one other bundle w.p.~$O(\log n / \sqrt{n})$ independently. Using the PDF of the binomial distribution it is easy to see that w.p.~$1-O(n^{-3})$ no more than $O(1)$ bundles combine with each other throughout the same hop. Similarly one may see that in $O(\log n)$ hops no fixed bundle merges more than $O(1)$ times. 
	
	Finally fix a node $v \in V$ and consider some hop $i \geq E_2$. Now, $v$ is hit on the $i$-th hop by some fixed bundle w.p.~less than $1 / ((n-1) - |\Fin|) = O(1/n)$. As at most $O(\log n \sqrt{n})$ bundles exist in the system, it follows that at most $O(1)$ bundles hit $v$ in the same hop. Also, the probability that $v$ is hit at least by one bundle in hop $i$ is smaller than
	\begin{equation*}
		1 - \left(1 - O \left( \frac{1}{n} \right)\right)^{O(\log n \sqrt{n})} < O \left( \frac{\log n}{\sqrt{n}}\right).
	\end{equation*}
Independent from any previous hops, $v$ receives no packets w.p.~ $1-O(\log n / \sqrt{n})$, and $O(1)$ load w.p.~$1-O(n^{-3})$. Therefore, the total load $v$ receives may be majorized by $O(1) \cdot B(3 \log_{1/\alpha} n, O(\log n / \sqrt{n}))$ with expected value $o(1)$. Looking again at the PDF yields that this load lies in $O(1)$ w.p.~$1-O(n^{-3})$, and applying the union bound we obtain that no node receives more than $O(1)$ bundles in total.
	
	Putting everything together we have that \begin{enumerate*} \item no bundle combines enough times to exceed size of $O(\sqrt{\log n})$, and \item no node is visited by more than $O(1)$ bundles in total \end{enumerate*}.
\end{proof}

%% file: main.bbl
\begin{thebibliography}{10}
\providecommand{\url}[1]{#1}
\csname url@samestyle\endcsname
\providecommand{\newblock}{\relax}
\providecommand{\bibinfo}[2]{#2}
\providecommand{\BIBentrySTDinterwordspacing}{\spaceskip=0pt\relax}
\providecommand{\BIBentryALTinterwordstretchfactor}{4}
\providecommand{\BIBentryALTinterwordspacing}{\spaceskip=\fontdimen2\font plus
\BIBentryALTinterwordstretchfactor\fontdimen3\font minus
  \fontdimen4\font\relax}
\providecommand{\BIBforeignlanguage}[2]{{%
\expandafter\ifx\csname l@#1\endcsname\relax
\typeout{** WARNING: IEEEtran.bst: No hyphenation pattern has been}%
\typeout{** loaded for the language `#1'. Using the pattern for}%
\typeout{** the default language instead.}%
\else
\language=\csname l@#1\endcsname
\fi
#2}}
\providecommand{\BIBdecl}{\relax}
\BIBdecl

\bibitem{ipfrr}
A.~Atlas and A.~Zinin, ``Basic specification for {IP} fast reroute: Loop-free
  alternates,'' in \emph{Request for Comments (RFC) 5286}, 2008.

\bibitem{mplsfrr}
P.~Pan, G.~Swallow, and A.~Atlas, ``Fast reroute extensions to {RSVP-TE} for
  {LSP} tunnels,'' in \emph{Request for Comments (RFC) 4090}, 2005.

\bibitem{offrr}
\BIBentryALTinterwordspacing
{{Switch Specification 1.3.1}}, ``{OpenFlow},'' 2013. [Online]. Available:
  \url{https://bit.ly/2VjOO77}
\BIBentrySTDinterwordspacing

\bibitem{srfrr}
P.~Fran{\c{c}}ois, C.~Filsfils, A.~Bashandy, B.~Decraene, and S.~Litkowski,
  ``Topology independent fast reroute using segment routing,'' 2014.

\bibitem{isis}
ISO, ``Intermediate ststem-to-intermediate system (is-is) routing protocol,''
  ISO/IEC 10589, 2002.

\bibitem{rfc2328}
J.~Moy, ``{OSPF} version 2,'' {RFC} editor,
  https://tools.ietf.org/html/rfc2328, {RFC} 2328, 1998.

\bibitem{filsfils2011bgp}
C.~Filsfils, P.~Mohapatra, J.~Bettink, P.~Dharwadkar, P.~De~Vriendt, Y.~Tsier,
  V.~Van Den~Schrieck, O.~Bonaventure, P.~Francois \emph{et~al.}, ``{BGP}
  prefix independent convergence,'' \emph{Cisco, Tech. Rep}, 2011.

\bibitem{kabbani2014flowbender}
A.~Kabbani, B.~Vamanan, J.~Hasan, and F.~Duchene, ``Flowbender: Flow-level
  adaptive routing for improved latency and throughput in datacenter
  networks,'' in \emph{Proc. of the 10th ACM International Conf. on emerging
  Networking Experiments and Technologies}, 2014, pp. 149--160.

\bibitem{franccois2014topology}
P.~Fran{\c{c}}ois, C.~Filsfils, A.~Bashandy, B.~Decraene, S.~Litkowski
  \emph{et~al.}, ``Topology independent fast reroute using segment routing,''
  2014.

\bibitem{gill2011understanding}
P.~Gill, N.~Jain, and N.~Nagappan, ``Understanding network failures in data
  centers: measurement, analysis, and implications,'' \emph{ACM SIGCOMM CCR},
  vol.~41, pp. 350--361, 2011.

\bibitem{concur2}
A.~K. Atlas and A.~Zinin, ``Basic specification for {IP} fast-reroute:
  loop-free alternates,'' \emph{IETF RFC 5286}, 2008.

\bibitem{elhourani2014ip}
T.~Elhourani, A.~Gopalan, and S.~Ramasubramanian, ``{IP} fast rerouting for
  multi-link failures,'' in \emph{Proc. IEEE INFOCOM}, 2014.

\bibitem{podc-ba}
J.~Feigenbaum, B.~Godfrey, A.~Panda, M.~Schapira, S.~Shenker, and A.~Singla,
  ``Brief announcement: On the resilience of routing tables,'' in \emph{Proc.
  ACM PODC}, 2012.

\bibitem{foerster2020feasibility}
K.-T. Foerster, J.~Hirvonen, Y.-A. Pignolet, S.~Schmid, and G.~Tredan, ``On the
  feasibility of perfect resilience with local fast failover,'' in \emph{Proc.
  SIAM APOCS}, 2021.

\bibitem{opodis13shoot}
M.~Borokhovich and S.~Schmid, ``How (not) to shoot in your foot with sdn local
  fast failover: A load-connectivity tradeoff,'' in \emph{Proc. OPODIS}, 2013.

\bibitem{borokhovich2018load}
M.~Borokhovich, Y.-A. Pignolet, S.~Schmid, and G.~Tredan, ``Load-optimal local
  fast rerouting for dense networks,'' \emph{IEEE/ACM Transactions on
  Networking}, vol.~26, no.~6, pp. 2583--2597, 2018.

\bibitem{wu2012ictcp}
H.~Wu, Z.~Feng, C.~Guo, and Y.~Zhang, ``Ictcp: Incast congestion control for
  tcp in data-center networks,'' \emph{IEEE/ACM transactions on networking},
  vol.~21, no.~2, pp. 345--358, 2012.

\bibitem{handley2017re}
M.~Handley, C.~Raiciu, A.~Agache, A.~Voinescu, A.~W. Moore, G.~Antichi, and
  M.~W{\'o}jcik, ``Re-architecting datacenter networks and stacks for low
  latency and high performance,'' in \emph{Proc. of ACM SIGCOMM}, 2017, pp.
  29--42.

\bibitem{infocom19casa}
K.-T. Foerster, Y.-A. Pignolet, S.~Schmid, and G.~Tredan, ``Casa: Congestion
  and stretch aware static fast rerouting,'' in \emph{Proc. IEEE INFOCOM},
  2019.

\bibitem{icnp19}
G.~Bankhamer, R.~Elsaesser, and S.~Schmid, ``Local fast rerouting with low
  congestion: A randomized approach,'' in \emph{Proc. 27th IEEE International
  Conference on Network Protocols (ICNP)}, 2020.

\bibitem{gravitymodel}
M.~Roughan, ``Simplifying the synthesis of internet traffic matrices,''
  \emph{ACM SIGCOMM CCR}, vol.~35, no.~5, p. 93–96, 2005.

\bibitem{clad2014disruption}
F.~Clad, ``Disruption-free routing convergence: computing minimal link-state
  update sequences,'' Ph.D. dissertation, Strasbourg, 2014.

\bibitem{clos}
M.~Al-Fares, A.~Loukissas, and A.~Vahdat, ``A scalable, commodity data center
  network architecture,'' \emph{ACM SIGCOMM CCR}, vol.~38, no.~4, pp. 63--74,
  2008.

\bibitem{singh2015jupiter}
A.~Singh, J.~Ong, A.~Agarwal, G.~Anderson, A.~Armistead, R.~Bannon, S.~Boving,
  G.~Desai, B.~Felderman, P.~Germano \emph{et~al.}, ``Jupiter rising: A decade
  of clos topologies and centralized control in google's datacenter network,''
  \emph{ACM SIGCOMM CCR}, vol.~45, no.~4, pp. 183--197, 2015.

\bibitem{DBLP:journals/ton/ChiesaNMGMSS17}
M.~Chiesa, I.~Nikolaevskiy, S.~Mitrovic, A.~V. Gurtov, A.~Madry, M.~Schapira,
  and S.~Shenker, ``On the resiliency of static forwarding tables,''
  \emph{IEEE/ACM Trans. Netw. (TON)}, vol.~25, pp. 1133--1146, 2017.

\bibitem{edmonds1973edge}
J.~Edmonds, ``Edge-disjoint branchings,'' \emph{Combinatorial algorithms},
  vol.~9, no. 91-96, p.~2, 1973.

\bibitem{bhalgat2008fast}
A.~Bhalgat, R.~Hariharan, T.~Kavitha, and D.~Panigrahi, ``Fast edge splitting
  and {Edmonds'} arborescence construction for unweighted graphs,'' in
  \emph{Proc. SODA}, 2008.

\bibitem{Chiesa2014}
\BIBentryALTinterwordspacing
M.~Chiesa, A.~Gurtov, A.~Madry, S.~Mitrovic, I.~Nikolaevkiy, A.~Panda,
  M.~Schapira, and S.~Shenker, ``Exploring the limits of static failover
  routing,'' 2014. [Online]. Available: \url{http://arxiv.org/abs/1409.0034}
\BIBentrySTDinterwordspacing

\bibitem{GareyJ79}
M.~R. Garey and D.~S. Johnson, \emph{Computers and Intractability: A Guide to
  the Theory of NP-Completeness}.\hskip 1em plus 0.5em minus 0.4em\relax W. H.
  Freeman, 1979.

\bibitem{frr-survey}
M.~Chiesa, A.~Kamisinski, J.~Rak, G.~Retvari, and S.~Schmid, ``A survey of
  fast-recovery mechanisms in packet-switched networks,'' \emph{IEEE
  Communications Surveys and Tutorials (COMST)}, 2021.

\bibitem{link-failures-ip-backbone}
G.~Iannaccone, C.-n. Chuah, R.~Mortier, S.~Bhattacharyya, and C.~Diot,
  ``Analysis of link failures in an {IP} backbone,'' in \emph{Proc. ACM SIGCOMM
  Workshop on Internet Measurment}, 2002.

\bibitem{failures-uninett}
A.~J. Gonz\'{a}lez and B.~E. Helvik, ``Analysis of failures characteristics in
  the {UNINETT IP} backbone network,'' in \emph{{IEEE} Inter. Conf. on Advanced
  Information Networking and Applications Workshops}, 2011.

\bibitem{ensure-dconn-nsdi13}
J.~Liu, A.~Panda, A.~Singla, B.~Godfrey, M.~Schapira, and S.~Shenker,
  ``Ensuring connectivity via data plane mechanisms,'' in \emph{Proc. NSDI},
  2013.

\bibitem{Yang14}
B.~Yang, J.~Liu, S.~Shenker, J.~Li, and K.~Zheng, ``{Keep Forwarding: Towards
  k-link failure resilient routing},'' in \emph{Proc. IEEE INFOCOM}, April
  2014, pp. 1617--1625.

\bibitem{gafni-lr}
E.~Gafni and D.~Bertsekas, ``Distributed algorithms for generating loop-free
  routes in networks with frequently changing topology,'' \emph{Trans.
  Commun.}, vol.~29, no.~1, pp. 11--18, 1981.

\bibitem{welch2011link}
J.~L. Welch and J.~E. Walter, ``Link reversal algorithms,'' \emph{Synthesis
  Lectures on Distributed Computing Theory}, vol.~2, pp. 1--103, 2011.

\bibitem{chiesa2020fast}
M.~Chiesa, A.~Kamisi{\'n}ski, J.~Rak, G.~R{\'e}tv{\'a}ri, and S.~Schmid, ``Fast
  recovery mechanisms in the data plane,'' 2020.

\bibitem{menth2009resilience}
M.~Menth, M.~Duelli, R.~Martin, and J.~Milbrandt, ``Resilience analysis of
  packet-switched communication networks,'' \emph{IEEE/ACM Transactions on
  Networking}, vol.~17, no.~6, pp. 1950--1963, 2009.

\bibitem{ref28}
S.~Kini, S.~Ramasubramanian, A.~Kvalbein, and A.~F. Hansen, ``Fast recovery
  from dual-link or single-node failures in ip networks using tunneling,''
  \emph{IEEE/ACM Trans. Netw. (TON)}, vol.~18, no.~6, pp. 1988--1999, 2010.

\bibitem{cohen2009maximizing}
R.~Cohen and G.~Nakibly, ``Maximizing restorable throughput in mpls networks,''
  \emph{IEEE/ACM Trans. Netw. (TON)}, vol.~18, no.~2, pp. 568--581, 2009.

\bibitem{qiu2010local}
J.~Qiu, M.~Gurusamy, K.~C. Chua, and Y.~Liu, ``Local restoration with multiple
  spanning trees in metro ethernet networks,'' \emph{IEEE/ACM Transactions On
  Networking}, vol.~19, no.~2, pp. 602--614, 2010.

\bibitem{r38}
D.~Wang and G.~Li, ``Efficient distributed bandwidth management for {MPLS} fast
  reroute,'' \emph{IEEE/ACM Trans. Netw. (TON)}, 2008.

\bibitem{ref11}
K.-W. Kwong, L.~Gao, R.~Gu{\'e}rin, and Z.-L. Zhang, ``On the feasibility and
  efficacy of protection routing in ip networks,'' \emph{IEEE/ACM Transactions
  on Networking (TON)}, vol.~19, no.~5, pp. 1543--1556, 2011.

\bibitem{clad2013graceful}
F.~Clad, P.~M{\'e}rindol, J.-J. Pansiot, P.~Francois, and O.~Bonaventure,
  ``Graceful convergence in link-state ip networks: A lightweight algorithm
  ensuring minimal operational impact,'' \emph{IEEE/ACM Trans. Netw. (TON)},
  vol.~22, no.~1, pp. 300--312, 2013.

\bibitem{kvalbein2008multiple}
A.~Kvalbein, A.~F. Hansen, T.~Cicic, S.~Gjessing, and O.~Lysne, ``Multiple
  routing configurations for fast ip network recovery,'' \emph{IEEE/ACM Trans.
  Netw. (TON)}, vol.~17, no.~2, pp. 473--486, 2008.

\bibitem{cho2011independent}
S.~Cho, T.~Elhourani, and S.~Ramasubramanian, ``Independent directed acyclic
  graphs for resilient multipath routing,'' \emph{IEEE/ACM Trans. Netw. (TON)},
  vol.~20, no.~1, pp. 153--162, 2011.

\bibitem{ref27}
A.~Gopalan and S.~Ramasubramanian, ``Multipath routing and dual link failure
  recovery in ip networks using three link-independent trees,'' in \emph{IEEE
  ANTS}, 2011, pp. 1--6.

\bibitem{elhourani2016ip}
T.~Elhourani, A.~Gopalan, and S.~Ramasubramanian, ``Ip fast rerouting for
  multi-link failures,'' \emph{IEEE/ACM Transactions on Networking}, vol.~24,
  no.~5, pp. 3014--3025, 2016.

\bibitem{fcp}
K.~Lakshminarayanan, M.~Caesar, M.~Rangan, T.~Anderson, S.~Shenker, and
  I.~Stoica, ``Achieving convergence-free routing using failure-carrying
  packets,'' \emph{ACM SIGCOMM CCR}, vol.~37, no.~4, pp. 241--252, 2007.

\bibitem{plinko-full}
B.~Stephens, A.~L. Cox, and S.~Rixner, ``Scalable multi-failure fast failover
  via forwarding table compression,'' in \emph{Proc ACM SOSR}, 2016.

\bibitem{robroute16infocom}
M.~Chiesa, I.~Nikolaevskiy, S.~Mitrovic, A.~Panda, A.~Gurtov, A.~Madry,
  M.~Schapira, and S.~Shenker, ``The quest for resilient (static) forwarding
  tables,'' in \emph{Proc. IEEE INFOCOM}, 2016.

\bibitem{francois2005evaluation}
P.~Francois and O.~Bonaventure, ``An evaluation of ip-based fast reroute
  techniques,'' in \emph{Proc. of the ACM conference on emerging network
  experiment and technology}, 2005, pp. 244--245.

\bibitem{icalp16}
M.~Chiesa, A.~V. Gurtov, A.~Madry, S.~Mitrovic, I.~Nikolaevskiy, M.~Schapira,
  and S.~Shenker, ``On the resiliency of randomized routing against multiple
  edge failures,'' in \emph{Proc. ICALP}, 2016.

\bibitem{dsn19}
K.-T. Foerster, A.~Kamisinski, Y.-A. Pignolet, S.~Schmid, and G.~Tredan,
  ``Bonsai: Efficient fast failover routing using small arborescences,'' in
  \emph{Proc. 49th IEEE/IFIP International Conference on Dependable Systems and
  Networks (DSN)}, 2019.

\bibitem{BES21}
G.~Bankhamer, R.~Elsaesser, and S.~Schmid, ``Randomized local fast rerouting
  for datacenter networks with almost optimal congestion,'' in \emph{Proc. 35th
  International Symposium on Distributed Computing (DISC)}, 2021, pp.
  9:1--9:19.

\bibitem{RS98}
M.~Raab and A.~Steger, ``{Balls into bins} -- {A} simple and tight analysis,''
  in \emph{Randomization and Approximation Techniques in Computer
  Science}.\hskip 1em plus 0.5em minus 0.4em\relax Springer Berlin Heidelberg,
  1998, pp. 159--170.

\bibitem{Drr2020}
B.~Doerr, \emph{Probabilistic Tools for the Analysis of Randomized Optimization
  Heuristics}.\hskip 1em plus 0.5em minus 0.4em\relax Springer, 2020.

\bibitem{motwani-book}
R.~Motwani and P.~Raghavan, \emph{Randomized Algorithms}.\hskip 1em plus 0.5em
  minus 0.4em\relax Cambridge University Press, 1995.

\bibitem{JP83}
K.~Joag-Dev and F.~Proschan, ``Negative association of random variables with
  applications,'' \emph{Ann. Statist.}, vol.~11, no.~1, pp. 286--295, 1983.

\bibitem{DP09}
D.~P. Dubhashi and A.~Panconesi, \emph{Concentration of Measure for the
  Analysis of Randomized Algorithms}.\hskip 1em plus 0.5em minus 0.4em\relax
  Cambridge University Press, 2009.

\bibitem{AD11}
A.~Auger and B.~Doerr, \emph{Theory of Randomized Search Heuristics:
  Foundations and Recent Developments}.\hskip 1em plus 0.5em minus 0.4em\relax
  World Scientific Publ., 2011.

\bibitem{simsource}
``Simulation source code,'' \url{https://github.com/gbank/CLOS-Simulations}.

\end{thebibliography}
